\documentclass{scrartcl}
\def\arXiv{}

\usepackage{basicstuff}

\usepackage{caption} %
\usepackage{subcaption} %
\usepackage[absolute,overlay]{textpos} %
\usepackage{tikz} %
\usepackage{textcomp} %
\usepackage{refcount} %
\usepackage{array}

\usepackage{booktabs}

\renewcommand{\todo}[1]{}

\newcommand*\circled[1]{ %
  \protect\tikz[baseline=(char.base)]{ %
    \protect\node[shape=circle,draw,inner sep=0.2pt] (char) {#1};}} %

\newcommand{\1}[1]{{\normalfont \ensuremath{#1^{\tiny\circled{1}}}}} %
\newcommand{\2}[1]{{\normalfont \ensuremath{#1^{\tiny\circled{2}}}}} %
\renewcommand{\k}[1]{{\normalfont \ensuremath{#1^{\tiny\circled{k}}}}} %
\newcommand{\proj}[2]{\ensuremath{\left.#1\right|_{#2}}} %
\newcommand{\eps}{\varepsilon}
\DeclareMathOperator{\skel}{skel} %
\DeclareMathOperator{\expan}{exp} %
\DeclareMathOperator{\pos}{pos} %
\DeclareMathOperator{\bel}{bel} %
\DeclareMathOperator{\roo}{root} %
\DeclareMathOperator{\cyc}{cyc} %
\DeclareMathOperator{\high}{high} %
\DeclareMathOperator{\lca}{LCA} %
\DeclareMathOperator{\contr}{contr} %
\DeclareMathOperator{\side}{side} %
\DeclareMathOperator{\en}{end} %
\DeclareMathOperator{\curr}{curr} %
\DeclareMathOperator{\prev}{prev} %
\DeclareMathOperator{\nex}{next} %
\DeclareMathOperator{\temp}{temp} %
\DeclareMathOperator{\detcyc}{detcyc} %

\ifdefined\elsevier

\usepackage{lmodern} %
\usepackage{amsthm} %

\theoremstyle{plain} %
\newtheorem{theorem}{Theorem} %
\newcounter{lemmacounter} %
\newtheorem{lemma}[lemmacounter]{Lemma} %
\newtheorem{corollary}{Corollary} %
\theoremstyle{definition} %

\fi

\ifdefined\arXiv

\usepackage{arXiv} 

\title{\Large Disconnectivity and Relative Positions\\in Simultaneous
  Embeddings\thanks{Part of this work was done within GRADR --
    EUROGIGA project no. 10-EuroGIGA-OP-003.}}

\author{Thomas Bläsius \myand Ignaz Rutter}

\date{Karlsruhe Institute of Technology (KIT), Germany\\%
  \medskip \texttt{\{blaesius,rutter\}@kit.edu}}

\fi

\begin{document}

\ifdefined\elsevier

\begin{frontmatter}
  
  \title{Disconnectivity and Relative Positions in Simultaneous
    Embeddings\tnoteref{eurogigia,confversion}}

  \tnotetext[eurogiga]{Part of this work was done within GRADR --
    EUROGIGA project no. 10-EuroGIGA-OP-003.} 

  \tnotetext[confversion]{A preliminary version of this work has been
    published as T. Bläsius and I. Rutter: Disconnectivity and Relative
    Positions in Simultaneous Embeddings, \emph{Proceedings of the 20th
      International Symposium on Graph Drawing (GD'12)}, pages 31--42,
    LNCS, Springer, 2012}

  \author{Thomas Bläsius} %
  \ead{blaesius@kit.edu}
  \author{Ignaz Rutter}
  \ead{rutter@kit.edu}
  \address{Karlsruhe Institute of Technology (KIT), Germany}
\fi

\ifdefined\arXiv
\maketitle
\fi

  \begin{abstract}
    For two planar graphs $\1G = (\1V, \1E)$ and $\2G = (\2V, \2E)$
    sharing a common subgraph $G = \1G \cap \2G$ the problem {\sc
      Simultaneous Embedding with Fixed Edges (SEFE)} asks whether
    they admit planar drawings such that the common graph is drawn the
    same.  Previous algorithms only work for cases where~$G$ is
    connected, and hence do not need to handle relative positions of
    connected components.  We consider the problem where~$G$, $\1G$
    and~$\2G$ are not necessarily connected.


    First, we show that a general instance of {\sc SEFE} can be
    reduced in linear time to an equivalent instance where $\1V = \2V$
    and $\1G$ and $\2G$ are connected.  Second, for the case where~$G$
    consists of disjoint cycles, we introduce the \emph{CC-tree} which
    represents all embeddings of~$G$ that extend to planar embeddings
    of~$\1G$.  We show that CC-trees can be computed in linear time,
    and that their intersection is again a CC-tree.  This yields a
    linear-time algorithm for {\sc SEFE} if all $k$ input graphs
    (possibly $k>2$) pairwise share the same set of disjoint cycles.
    These results, including the CC-tree, extend to the case where $G$
    consists of arbitrary connected components, each with a fixed
    planar embedding on the sphere.  Then the running time
    is~$O(n^2)$.
  \end{abstract}

\ifdefined\elsevier
\end{frontmatter}
\fi

\section{Introduction}
\label{sec:introduction}

To enable a human reader to compare different relational datasets on a
common set of objects it is important to visualize the corresponding
graphs in such a way that the common parts of the different datasets
are drawn as similarly as possible.  An example is a dynamic graph
that changes over time.  Then the change between two points in time
can be easily grasped with the help of a visualization showing the
parts that did not change in the same way for both graphs.  This leads
to the fundamental theoretical problem {\sc Simultaneous Embedding
  with Fixed Edges} (or {\sc SEFE} for short), asking for two graphs
$\1G = (\1V, \1E)$ and $\2G = (\2V, \2E)$ with the common graph $G =
(V, E) = (\1V \cap \2V, \1E \cap \2E)$, whether there are planar
drawings of $\1G$ and $\2G$ such that the common graph $G$ is drawn
the same in both.

The problem {\sc SEFE} and its variants, such as {\sc Simultaneous
  Geometric Embedding}, where one insists on a simultaneous
straight-line drawing, have been studied intensively in the past
years; see the recent survey~\cite{bkr-sepg-13} for an overview.  Some
of the results show, for certain graph classes, that they always admit
simultaneous embeddings or that there exist negative instances of SEFE
whose input graphs belong to these classes.  As there are planar
graphs that cannot be embedded simultaneously, the question of
deciding whether given graphs admit a {\sc SEFE} is of high interest.
Gassner et al.~\cite{SimultaneousGraphEmbeddings-Gassner.etal(06)}
show that it is $\mathcal {NP}$-complete to decide {\sc SEFE} for
three or more graphs.  For two graphs the complexity status is still
open.  However, there are several approaches yielding efficient
algorithms for special cases.  Fowler at al. show how to solve {\sc
  SEFE} efficiently, if~$\1G$ and~$G$ have at most two and one cycles,
respectively~\cite{SPQR-TreeApproachto-Fowler.etal(09)}.  Fowler et
al. characterize the class of common graphs that always admit a {\sc
  SEFE}~\cite{Characterizationsofrestricted-Fowler.etal(11)}.
Angelini et al.~\cite{adf-tppeg-10} show that if one of the input
graphs has a fixed planar embedding, then {\sc SEFE} can be solved in
linear time.  Haeupler et al. solve {\sc SEFE} in linear time for the
case that the common graph is
biconnected~\cite{TestingSimultaneousPlanarity-Haeupler.etal(10)}.
Angelini et al. obtain the same result with a completely different
approach~\cite{adfpr-tsegi-12}.  They additionally solve the case
where the common graph is a star and, moreover, show the equivalence
of the case where the common graph is connected to the case where the
common graph is a tree and relate it to a constrained book embedding
problem.  The currently least restrictive result (in terms of
connectivity) by Bläsius and Rutter~\cite{br-spqoacep-13} shows that
{\sc SEFE} can be solved in polynomial time for the case that both
graphs are biconnected and the common graph is connected.

The algorithms testing {\sc SEFE} have in common that they use the
result by Jünger and
Schulz~\cite{IntersectionGraphsin-Juenger.Schulz(09)} stating that the
question of finding a simultaneous embedding for two graphs is
equivalent to the problem of finding planar embeddings of $\1G$ and
$\2G$ such that they induce the same embedding on~$G$.  Moreover, they
have in common that they all assume that the common graph is
connected, implying that it is sufficient to enforce the common edges
incident to each vertex to have the same circular ordering in both
embeddings.  Especially in the result by Bläsius and
Rutter~\cite{br-spqoacep-13} this is heavily used, as they explicitly
consider only orders of edges around vertices using PQ-trees.
However, if the common graph is not required to be connected, we
additionally have to care about the relative positions of connected
components to one another, which introduces an additional difficulty.
Note that the case where the common graph is disconnected cannot be
reduced to the case where it is connected by inserting additional
edges.  Figure~\ref{fig:cannot-connect-common-graph} shows an instance
that admits a simultaneous embedding, which is no longer true if the
isolated vertex $v$ is connected to the remaining graph.  Other
approaches to solve the SEFE problem have only appeared recently.
Schaefer~\cite{s-ttp-13} characterizes, for certain classes of SEFE
instances, the pairs of graphs that admit a SEFE via the independent
odd crossing number.  Among others, this gives a polynomial-time
algorithm for SEFE when the common graph has maximum degree~3 and is
not necessarily connected.

\begin{figure}
  \centering
  \includegraphics{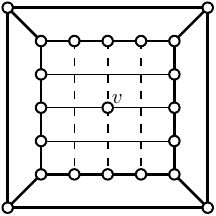}
  \caption{The bold edges belong to both graphs, the dashed and thin
    edges are exclusive edges.}
  \label{fig:cannot-connect-common-graph}
\end{figure}

In this work we tackle the {\sc SEFE} problem from the opposite
direction \todo{Pag 3, line -6}than the so far known results, by
assuming that the circular order of edges around vertices in $G$ is
already fixed and we only have to ensure that the embeddings chosen
for the input graphs are \emph{compatible} in the sense that they
induce the same relative positions on~$G$.  Initially, we assume that
the graph $G$ consists of a set of disjoint cycles, each of them
having a unique planar embedding.  We present a novel data structure,
the \emph{CC-tree}, which represents all embeddings of a set of
disjoint cycles that can be induced by an embedding of a graph
containing them as a subgraph.  We moreover show that two such
CC-trees can be intersected, again yielding a CC-tree.  Thus, for the
case that $\1G$ and $\2G$ have the common graph $G$ consisting of a
set of disjoint cycles, the intersection of the CC-trees corresponding
to $\1G$ and $\2G$ represents all simultaneous embeddings.  We show
that CC-trees can be computed and intersected in linear time, yielding
a linear-time algorithm to solve {\sc SEFE} for the case that the
common graph consists of disjoint cycles.  Note that this obviously
also yields a linear-time algorithm to solve {\sc SEFE} for more than
two graphs if they all share the same common graph consisting of a set
of disjoint cycles.  We show that these results can be further
extended to the case where the common graph may contain arbitrary
connected components, each of them with a prescribed planar embedding.
However, in this case the corresponding data structure, called
CC$^\oplus$-tree, may have quadratic size.  These results show that
the choice of relative positions of several connected components does
not solely make the problem {\sc SEFE} hard to solve.

\todo{Pag 3, line 14-15:}Note that these results have an interesting
application concerning the problem {\sc Partially Embedded Planarity}.
The input of \textsc{Partially Embedded Planarity} is a planar graph
$G$ together with a fixed embedding for a subgraph $H$ (including
fixed relative positions).  It asks whether $G$ admits a planar
embedding extending the embedding of $H$.  Angelini et
al.~\cite{adf-tppeg-10} introduced this problem and solve it in linear
time.  The CC$^\oplus$-tree can be used to solve {\sc Partially
  Embedded Planarity} in quadratic time as it represents all possible
relative positions of the connected components in $H$ to one another
that can be induced by an embedding of $G$.  It is \todo{Pag 3, line
  18:}then easy to test whether the prespecified relative positions
can be achieved.  In fact, this solves the slightly more general case
of {\sc Partially Embedded Planarity} where not all relative positions
have to be fixed.

The above described results have one restriction that was not
mentioned so far.  The graphs~$\1G$ and $\2G$ are assumed to be
connected, otherwise the approach we present does not work.
Fortunately, we can show that both graphs of an instance of {\sc SEFE}
can always be assumed to be connected, even if all vertices are
assumed to be common vertices (forming isolated vertices when not
connected via a common edge).  This shows that {\sc SEFE} can be solved
efficiently if the common graph consists of disjoint cycles without
further restrictions on the connectivity.  Moreover, it is an
interesting result on its own as it applies to arbitrary instances of
{\sc SEFE}, not only to the special case we primarily consider here.

As connectivity plays an important role in this work we fix some basic
definitions in the following.  A graph is \emph{connected} if there
exists a path between any pair of vertices.  A \emph{separating
  $k$-set} is a set of $k$ vertices whose removal disconnects the
graph.  Separating 1-sets and 2-sets are \emph{cutvertices} and
\emph{separation pairs}, respectively.  A connected graph is
\emph{biconnected} if it does not have a cut vertex and
\emph{triconnected} if it does not have a separation pair.  The
maximal biconnected components of a graph are called \emph{blocks}.
The \emph{cut components} with respect to a separating $k$-set $S$ are
the maximal subgraphs that are not disconnected by removing $S$.

\paragraph{Outline.}

In Section~\ref{sec:disconnected-graphs} we show that, for any given
instance of {\sc SEFE}, there exists an equivalent instance such that
both input graphs are connected, even if each vertex is assumed to be
a common vertex.  With this result instances of {\sc SEFE} can always
be assumed to have this property.  In Section~\ref{sec:disj-cycl} we
show how to solve {\sc SEFE} in linear time if the common graph
consists of disjoint cycles, including a compact representation of all
simultaneous embeddings.  In Section~\ref{sec:extension} we show how
to extend these results to solve {\sc SEFE} in quadratic time for the
case that the common graph consists of arbitrary connected components,
each with a fixed planar embedding.  We conclude in
Section~\ref{sec:conclusion}.

\section{Connecting Disconnected Graphs}
\label{sec:disconnected-graphs}

Let $\1G = (V, \1E)$ and $\2G = (V, \2E)$ be two planar graphs with
common graph $G = (V, E)$ with $E = \1E \cap \2E$.  We show that the
problem {\sc SEFE} can be reduced to the case where $\1G$ and $\2G$
are required to be connected.  First note that the connected
components of the union of $\1G$ and $\2G$ can be handled
independently.  Thus we can assume that $\1G \cup \2G$ is connected.
We first ensure that $\1G$ is connected without increasing the number
of connected components in $\2G$.  Afterwards we can apply the same
steps to $\2G$ to make it connected, maintaining the connectivity of
$\1G$.

Assume $\1G$ and $\2G$ consist of $\1k$ and $\2k$ connected
components, respectively.  Since the union of $\1G$ and $\2G$ is
connected, we can always find an edge $\2e = \{v_1, v_2\} \in \2E$ such
that the vertices $v_1$ and $v_2$ belong to different connected
components $\1{H_1}$ and $\1{H_2}$ in $\1G$.  We construct the
\emph{augmented instance} $(\1{G_+}, \2{G_+})$ of {\sc SEFE} with
respect to the edge $\2e$ by introducing a new vertex $v_{12}$ and new
edges $e = \{v_1, v_{12}\} \in E$ and $\1e = \{v_{12},v_2\} \in \1E$.
Note that $\1{G_+}$ has $\1k - 1$ connected components since $\1{H_1}$
and $\2{H_2}$ are now connected via the two edges $e$ and $\1e$.
Moreover, the number $\2k$ of connected components in $\2G$ does not
change, since the edge $e$ connects the new vertex $v_{12}$ to one of
its connected components.  It remains to show that the original
instance and the augmented instance are equivalent.

\begin{figure}
  \centering
  \includegraphics[page=1]{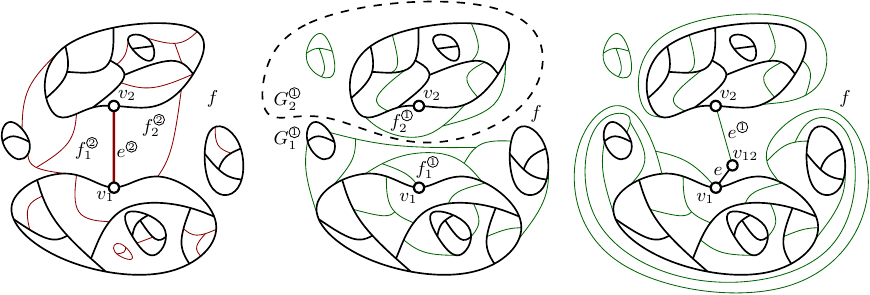}
  \caption{Illustration of Lemma~\ref{lem:connecting-graphs}, the
    common graph is depicted black.  The graph $\2G$ with the edge
    $\2e = \{v_1,v_2\}$ lying in the common face $f$, which is the
    outer face of $G$ (left).  The graph $\1G$ with the faces
    $\1{f_1},\1{f_2} \in \1{\mathcal F}(f)$ incident to $v_1$ and
    $v_2$, respectively, partitioned into $\1{G_1}$
    and~$\1{G_2}$~(middle). The resulting graph $\1G$ after choosing
    $\1{f_i}$ as outer face of $\1{G_i}$ (for $i = 1,2$) and inserting
    the vertex $v_{12}$ and the edges $e$ and $\1e$ (right).}
  \label{fig:connecting-graphs}
\end{figure}

\begin{lemma}
\label{lem:connecting-graphs}
  Let $(\1G, \2G)$ be an instance of {\sc SEFE} and let $(\1{G_+},
  \2{G_+})$ be the augmented instance with respect to the edge $\2e =
  \{v_1, v_2\}$.  Then $(\1G, \2G)$ and $(\1{G_+}, \2{G_+})$ are
  equivalent.
\end{lemma}
\begin{proof}
  If the augmented instance admits a {\sc SEFE}, then obviously the
  original instance does.  To show the other direction assume the
  original instance $(\1G, \2G)$ has a {\sc SEFE} $(\1{\mathcal E},
  \2{\mathcal E})$ inducing the embedding $\mathcal E$ for the common
  graph.  We show how to construct an embedding $\1{\mathcal E'}$ such
  that \todo{Pag 4, proof of Lemma 1, lines 2-3:}(i) $(\1{\mathcal
    E'}, \2{\mathcal E})$ is a {\sc SEFE}, and (ii) the vertices $v_1$
  and $v_2$ lie on the border of a common face in~$\1{\mathcal E'}$.
  Then we can easily add the vertex $v_{12}$ together with the two
  edges $e$ and~$\1e$, yielding a {\sc SEFE} of the augmented instance
  $(\1{G_+}, \2{G_+})$.  Note that the first property, namely that
  $(\1{\mathcal E'}, \2{\mathcal E})$ is a {\sc SEFE}, is satisfied if
  and only if the embeddings $\1{\mathcal E}$ and $\1{\mathcal E'}$
  induce the same embedding $\mathcal E$ on the common graph.
  Figure~\ref{fig:connecting-graphs} illustrates the proof.

  Consider a face $f$ of the embedding $\mathcal E$ of the common
  graph.  The embedding $\1{\mathcal E}$ of the graph~$\1G$ splits
  this face $f$ into a set of faces $\1{\mathcal F}(f) = \{\1{f_1},
  \dots, \1{f_k}\}$.  We say that a face $\1f \in \1{\mathcal F}(f)$
  is \emph{contained} in $f$.  Note that every face of $\1{\mathcal
    E}$ is contained in exactly one face of $\mathcal E$.  The same
  definition can be made for the second graph.

  The edge $\2e = \{v_1, v_2\}$ borders two faces $\2{f_1}$ and
  $\2{f_2}$ of $\2{\mathcal E}$.  Since $\2e$ belongs exclusively
  to~$\2G$ (otherwise $v_1$ and $v_2$ would not have been in different
  connected components in $\1G$) both faces~$\2{f_1}$ and $\2{f_2}$
  are contained in the same face $f$ of the embedding $\mathcal E$ of
  the common graph $G$.  We assume without loss of generality that $f$
  is the outer face.  The face $f$ may be subdivided by edges
  belonging exclusively to the graph $\1G$.  However, we can find
  faces $\1{f_1}$ and $\1{f_2}$ of $\1{\mathcal E}$, both contained
  in~$f$, such that $v_1$ and $v_2$ are contained in the boundary of
  these faces.  If $\1{f_1} = \1{f_2}$ we are done since $v_1$ and
  $v_2$ lie on the boundary of the same face in $\1{\mathcal E}$.
  Otherwise, we split $\1G$ into two subgraphs~$\1{G_1}$ and~$\1{G_2}$
  with the embeddings $\1{\mathcal E_1}$ and $\1{\mathcal E_2}$
  induced by $\1{\mathcal E}$ as follows.  The connected
  component~$\1{H_i}$ (for~$i = 1, 2$) containing $v_i$ belongs to
  $\1{G_i}$ and all connected components that are completely contained
  in an internal face of $\1{H_i}$ also belong to $\1{G_i}$.  All
  remaining connected components belong either to $\1{G_1}$ or to
  $\1{G_2}$.  Note that this partition ensures that there is a simple
  closed curve in the outer face of $\1{\mathcal E}$ separating
  $\1{G_1}$ and $\1{G_2}$.  Thus, we can change the embeddings
  of~$\1{\mathcal E_1}$ and~$\1{\mathcal E_2}$ independently.  In
  particular, we choose the faces $\1{f_1}$ and $\1{f_2}$ to be the
  new outer faces, yielding the changed embeddings $\1{\mathcal E_1'}$
  and $\1{\mathcal E_2'}$, respectively.  When combining these to
  embeddings by putting $\1{G_1}$ into the outer face of $\1{G_2}$ and
  vice versa, we obtain a new embedding $\1{\mathcal E'}$ of $\1G$
  with the following two properties.  First, the embedding induced for
  the common graph does not change since both faces $\1{f_1}$ and
  $\1{f_2}$ belong to the outer face $f$ of the embedding $\mathcal E$
  of the common graph $G$.  Second, the vertices $v_1$ and $v_2$ both
  lie on the outer face of the embedding $\1{\mathcal E'}$.  Hence,
  $(\1{\mathcal E'}, \2{\mathcal E})$ is still a {\sc SEFE} of the
  instance $(\1G, \2G)$ and the vertex $v_{12}$ together with the two
  edges $e$ and~$\1e$ can be added easily, which concludes the proof.
\end{proof}

With this construction we can reduce the number of connected
components of $\1G$ and $\2G$ and thus finally obtain an equivalent
instance of {\sc SEFE} in which both graphs are connected.  We obtain
the following Theorem.

\begin{theorem}
  \label{thm:connecting-graphs}
  For every instance $(\1G, \2G)$ of {\sc SEFE} there exits an
  equivalent instance $(\1{G_{++}}, \2{G_{++}})$ such that
  $\1{G_{++}}$ and $\2{G_{++}}$ are connected.  Such an instance can
  be computed in linear time.
\end{theorem}
\begin{proof}
  Lemma~\ref{lem:connecting-graphs} directly implies that an
  equivalent instance $(\1{G_{++}}, \2{G_{++}})$ in which both graphs
  are connected exists.  It remains to show that it can be computed in
  linear time.  To connect all the connected components of $\1G$, we
  contract each of them to a single vertex in the graph $\2G$.  Then
  an arbitrary spanning tree yields a set of edges $\2{e_1},
  \dots,\2{e_k} \in \2E$, such that augmenting the instance with
  respect to these edges yields a connected graph $\1{G_{++}}$.  This
  works symmetrically for $\2G$ and can obviously be done in linear
  time.
\end{proof}

\section{Disjoint Cycles}
\label{sec:disj-cycl}

In this section, we consider the problem {\sc SEFE} for the case that
the common graph consists of a set of disjoint cycles.  Due to
Theorem~\ref{thm:connecting-graphs}, we can assume without loss of
generality that both graphs are connected.  In
Section~\ref{sec:disj-cycl-poly-time} we show how to solve this
special case of {\sc SEFE} in polynomial time.  In
Section~\ref{sec:compact-rep} we introduce a tree-like data structure,
the \emph{CC-tree}, representing all planar embeddings of a set of
cycles contained in a single graph that can be induced by an embedding
of the whole graph.  We additionally show that the intersection of the
set of embeddings represented by two CC-trees can again be represented
by a CC-tree, yielding a solution for {\sc SEFE} even for the case of
more than two graphs if all graphs have the same common graph, which
consists of a set of disjoint cycles.  In
Section~\ref{sec:line-time-algor} we show how to compute the CC-tree
and the intersection of two CC-trees in linear time.  Before we start,
we fix some definitions.

\paragraph{Embeddings of Disjoint Cycles.}
\label{sec:embedd-disj-cycl}

Let $\mathcal C = \{C_1, \dots, C_k\}$ be a set of disjoint simple
cycles.  We consider embeddings of these cycles on the sphere.  Since
a single cycle has a unique embedding on the sphere only their
relative positions to one another are of interest.  To be able to use
the terms ``left'' and ``right'' we consider the cycles to be
directed.  We denote the relative position of a cycle $C_j$ with
respect to a cycle $C_i$ by $\pos_{C_i}(C_j)$.  More precisely, we
have $\pos_{C_i}(C_j) = \text{``left''}$ and $\pos_{C_i}(C_j) =
\text{``right''}$, if $C_j$ lies on the left and right side of $C_i$,
respectively.  We call an assignment of a value ``left'' or ``right''
to each of these relative positions a \emph{semi-embedding} of the
cycles $\mathcal C = \{C_1, \dots, C_k\}$.  Note that not every
semi-embedding yields an embedding of the cycles.  For example if
$\pos_{C_i}(C_j) = \pos_{C_j}(C_k) = \text{``left''}$ and
$\pos_{C_j}(C_i) = \text{``right''}$, then $\pos_{C_i}(C_k)$ also
needs to have the value ``left''; see Figure~\ref{fig:semi-embedding}.
However, two embeddings yielding the same semi-embedding are the same.

Sometimes we do not only consider the relative position of cycles but
also of some other disjoint subgraph.  We extend our notation to this
case.  For example the relative position of a single vertex $v$ with
respect to a cycle $C$ is denoted by $\pos_C(v)$.

\begin{figure}
  \centering
  \includegraphics{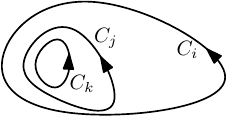}
  \caption{Three nested cycles.}
  \label{fig:semi-embedding}
\end{figure}

\paragraph{SPQR- and BC-Trees.}
\label{sec:spqr-trees}

The \emph{block-cutvertex tree (BC-tree)} $\mathcal B$ of a connected
graph is a tree whose nodes are the blocks and cutvertices of the
graph, called \emph{B-nodes} and \emph{C-nodes}, respectively.  In the
BC-tree a block $B$ and a cutvertex $v$ are joined by an edge if $v$
belongs to $B$.  If an embedding is chosen for each block, these
embeddings can be combined to an embedding of the whole graph if and
only if $\mathcal B$ can be rooted at a B-node such that the parent of
every other block $B$ in $\mathcal B$, which is a cutvertex, lies on
the outer face of $B$.

We use the \emph{SPQR-tree} introduced by Di Battista and
Tamassia~\cite{dt-omtc-96,dt-opt-96} to represent all planar
embeddings of a biconnected planar graph $G$.  The SPQR-tree $\mathcal
T$ of $G$ is a decomposition of $G$ into triconnected components along
its \emph{split pairs}, where a split pair is either a separation pair
or an edge.  We define the SPQR-tree to be unrooted, representing
embeddings on the sphere, that is planar embeddings without a
designated outer face.  Let $\{s, t\}$ be a split pair and let $H_1$
and~$H_2$ be two subgraphs of $G$ such that $H_1 \cup H_2 = G$ and
$H_1 \cap H_2 = \{s, t\}$.  Consider the tree containing the two nodes
$\mu_1$ and $\mu_2$ associated with the graphs $H_1 + \{s, t\}$ and
$H_2 + \{s, t\}$, respectively.  These graphs are called
\emph{skeletons} of the nodes $\mu_i$, denoted by $\skel(\mu_i)$ and
the special edge $\{s, t\}$ is said to be a \emph{virtual edge}.  The
two nodes $\mu_1$ and $\mu_2$ are connected by an edge or, more
precisely, the occurrence of the virtual edges $\{s, t\}$ in both
skeletons are linked by this edge.  The \emph{expansion graph}
$\expan(\{s, t\})$ of a virtual edge $\{s, t\}$ is the subgraph of $G$
it represents, that is in $\skel(\mu_1)$ and $\skel(\mu_2)$ the
expansion graphs of $\{s, t\}$ are $H_2$ and $H_1$, respectively.  Now
a combinatorial embedding of $G$ uniquely induces a combinatorial
embedding of $\skel(\mu_1)$ and $\skel(\mu_2)$.  Furthermore,
arbitrary and independently chosen embeddings for the two skeletons
determine an embedding of $G$, thus the resulting tree can be used to
represent all embeddings of $G$ by the combination of all embeddings
of two smaller planar graphs.  This replacement can of course be
applied iteratively to the skeletons yielding a tree with more nodes
but smaller skeletons associated with the nodes.

Applying this kind of decomposition in a systematic way yields the
SPQR-tree as introduced by Di Battista and
Tamassia~\cite{dt-omtc-96,dt-opt-96}.  The SPQR-tree $\mathcal T$ of a
biconnected planar graph $G$ contains four types of nodes.  \todo{Pag
  7, line 1-2:}First, the skeleton of a P-node consists of a bundle of
at least three parallel edges.  Embedding the skeleton of a P-node
corresponds to choosing an order for the parallel edges.  Second, the
skeleton of an R-node is triconnected, thus having exactly two
embeddings~\cite{w-cgcg-32}, and third, S-nodes have a simple cycle as
skeleton without any choice for the embedding.  Finally, every edge in
a skeleton representing only a single edge in the original graph $G$
is formally also considered to be a virtual edge linked to a Q-node in
$\mathcal T$ representing this single edge.  Note that all leaves of
the SPQR-tree~$\mathcal T$ are Q-nodes.  Besides from being a nice way
to represent all embeddings of a biconnected planar graph, the
SPQR-tree has size only linear in the size of $G$ and Gutwenger and
Mutzel~\cite{gm-lti-00} show that it can be computed in linear time.
Figure~\ref{fig:spqr-tree} shows a biconnected planar graph together
with its SPQR-tree.

\begin{figure}
  \centering
  \includegraphics[page=1]{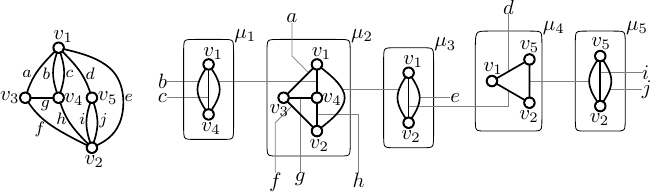}
  \caption{The unrooted SPQR-tree of a biconnected planar graph.  The
    nodes $\mu_1$, $\mu_3$ and $\mu_5$ are P-nodes, $\mu_2$ is an
    R-node and $\mu_4$ is an S-node.  The Q-nodes are not shown
    explicitely.}
  \label{fig:spqr-tree}
\end{figure}

\subsection{A Polynomial-Time Algorithm}
\label{sec:disj-cycl-poly-time}

Let $(\1G, \2G)$ be an instance of {\sc SEFE} with common graph $G$
consisting of pairwise disjoint simple cycles $\mathcal C = \{C_1,
\dots, C_k\}$.  We first assume that $\1G$ and $\2G$ are biconnected
and show later how to remove this restriction.  Our approach is to
formulate constraints on the relative positions of the cycles to one
another ensuring that $\1G$ and $\2G$ induce the same semi-embedding
on the common graph $G$.  We show implicitly that the resulting
semi-embedding is really an embedding by showing that the graphs $\1G$
and $\2G$ have planar embeddings inducing this semi-embedding.  Note
that this only works for the case that $\1G$ and $\2G$ are connected.
Thus, our approach crucially relies on the result provided in
Section~\ref{sec:disconnected-graphs}.

\subsubsection*{Biconnected Graphs}
\label{sec:biconnected-graphs}

Before considering two graphs, we determine for a single graph the
possible embeddings it may induce on a set of disjoint cycles
contained in it.  Let $G = (V, E)$ be a biconnected graph with
SPQR-tree $\mathcal T$, let $C$ be a simple directed cycle in $G$ and
let $\mu$ be a node in $\mathcal T$.  Obviously, $C$ is either
completely contained in the expansion graph of a single virtual edge
of $\mu$ or $C$ induces a simple directed cycle of virtual edges in
$\skel(\mu)$.  We say that $C$ is \emph{contracted} in $\skel(\mu)$ in
the first case and that $C$ is \emph{a cycle} in $\skel(\mu)$ in the
second case.
If~$C$
is a cycle in~$\skel(\mu)$, we also say that~$\skel(\mu)$
\emph{contains}~$C$ as a cycle.  Consider the case where $C$ is a
cycle in $\skel(\mu)$ and let~$\kappa$ denote this cycle.  By fixing
the embedding of $\skel(\mu)$ the virtual edges in $\skel(\mu)$ not
contained in $\kappa$ split into two groups, some lie to the left and
some to the right of $\kappa$.  Obviously, a vertex~$v \in V \setminus
V(C)$ in the expansion graph of a virtual edge that lies to the left
(to the right) of~$\kappa$ lies to the left (to the right) of $C$ in
$G$, no matter which embedding is chosen for the skeletons of other
nodes.  In other words, the value of $\pos_{C}(v)$ is completely
determined by this single node~$\mu$.  We show that for every vertex
$v \in V \setminus V(C)$ there is a node $\mu$ in~$\mathcal T$ containing $C$ as a
cycle such that the virtual edge in $\skel(\mu)$ containing $v$ in its
expansion graph is not contained in the cycle $\kappa$ induced by $C$.
Hence such a node~$\mu \in \mathcal T$ determining $\pos_{C}(v)$
always exists.  Extending this to a pair of cycles yields the
following lemma.

\begin{lemma}
  \label{lem:determining-pos-by-ex-one}
  Let $G$ be a biconnected planar graph with SPQR-tree $\mathcal T$
  and let $C_1$ and $C_2$ be two disjoint simple cycles in $G$.  There
  is exactly one node $\mu$ in $\mathcal T$ determining
  $\pos_{C_1}(C_2)$.  Moreover, $\mu$ contains $C_1$ as cycle
  $\kappa_1$ and $C_2$ either as a cycle or contracted in an edge not
  contained in $\kappa_1$.
\end{lemma}
\begin{proof}
  We choose some vertex $v \in V(C_2)$ as representative for the whole
  cycle.  Consider a Q-node~$\mu_1$ in the SPQR-tree $\mathcal T$
  corresponding to an edge contained in $C_1$.  Moreover, let $\mu_k$
  be a Q-node corresponding to an edge incident to $v$.  We claim that
  the desired node $\mu$ lies somewhere on the path $\mu_1, \dots,
  \mu_k$ in the SPQR-tree $\mathcal T$.

  Obviously $C_1$ is a cycle in $\mu_1$ and the vertex $v$ belongs to
  the virtual edge in $\skel(\mu_1)$.  In $\mu_k$ the vertex $v$ is a
  pole and $C_1$ is contracted in the virtual edge of $\skel(\mu_k)$
  since $v \notin V(C_1)$.  Assume we are navigating from $\mu_1$ to
  $\mu_k$ and let $\mu_i$ be the current node.  If $\skel(\mu_i)$ does
  not contain the vertex $v$, it belongs to a single virtual edge in
  $\skel(\mu_i)$.  In this case $\mu_{i+1}$ is obviously the node
  corresponding to this virtual edge.  If $v$ is a vertex of
  $\skel(\mu_i)$, then $\mu_{i+1}$ corresponds to one of the virtual
  edges incident to $v$ in $\skel(\mu_i)$.  As long as $C_1$ is a
  cycle in the current node and $v$ belongs to a virtual edge in this
  cycle, the next node in the path corresponds to this virtual edge
  and thus~$C_1$ remains a cycle.  Since $C_1$ is contracted in
  $\mu_k$, we somewhere need to follow a virtual edge not contained in
  the cycle induced by $C_1$; let $\mu$ be this node.  By definition
  $\mu$ contains $C_1$ as cycle~$\kappa$ and the next node on the path
  belongs to a virtual edge that is not contained in $\kappa$ but
  contains~$v$ in its expansion graph.  Thus $\pos_{C_1}(v)$ is
  determined by this node $\mu$.  Since $v$ is a node of the second
  cycle $C_2$ also $\pos_{C_1}(C_2)$ is completely determined by this
  node.  Moreover, $\mu$ contains $C_1$ as cycle~$\kappa$ and $C_2$
  either as a cycle or contracted in a virtual edge not belonging to
  $\kappa$.
\end{proof}

Now consider a set of pairwise disjoint cycles $\mathcal C = \{C_1,
\dots, C_k\}$ in $G$.  Let $\mu$ be an arbitrary node in the SPQR-tree
$\mathcal T$.  If $\mu$ is an S- or a Q-node it clearly does not
determine any of the relative positions since either every cycle is
contracted in $\skel(\mu)$ or a single cycle is a cycle in
$\skel(\mu)$ containing all the virtual edges.  In the following, we
consider the two interesting cases namely that $\mu$ is an R- or a
P-node containing at least one cycle as a cycle.

Let $\mu$ be a {\bf P-node} in $\mathcal T$ with $\skel(\mu)$
consisting of two vertices $s$ and $t$ with parallel virtual
edges~$\eps_1, \dots, \eps_\ell$ between them.  If $C \in \mathcal C$
is contained as a cycle in $\skel(\mu)$, it induces a cycle $\kappa$
in $\skel(\mu)$ consisting of two of the parallel virtual edges.  Let
without loss of generality $\eps_1$ and $\eps_2$ be these virtual
edges.  Obviously, no other cycle $C' \in \mathcal C$ is a cycle in
$\skel(\mu)$ since such a cycle would need to contain $s$ and $t$,
which is a contradiction to the assumption that $C$ and $C'$ are
disjoint.  Thus, every other cycle $C'$ is contracted in $\skel(\mu)$,
belonging to one of the virtual edges $\eps_1, \dots, \eps_\ell$.  If
it belongs to $\eps_1$ or $\eps_2$, which are contained in $\kappa$,
then $\pos_{C}(C')$ is not determined by~$\mu$.  If $C'$ belongs to
one of the virtual edges $\eps_3, \dots, \eps_\ell$, the relative
position $\pos_{C}(C')$ is determined by the relative position of this
virtual edge with respect to the cycle $\kappa$.  This relative
position can be chosen for every virtual edge $\eps_3, \dots,
\eps_\ell$ arbitrarily and independently.  Hence, if there are two
cycles $C_i$ and~$C_j$ belonging to different virtual edges in $\mu$, the
positions $\pos_{C}(C_i)$ and $\pos_{C}(C_j)$ can be chosen
independently.  Furthermore, if the two cycles $C_i$ and $C_j$ belong
to the same virtual edge $\eps \in \{\eps_3, \dots, \eps_\ell\}$,
their relative position with respect to $C$ is the same, that is
$\pos_{C}(C_i) = \pos_{C}(C_j)$, for every embedding of $G$.

Let $\mu$ be an {\bf R-node} in $\mathcal T$.  For the moment, we
consider that the embedding of $\skel(\mu)$ is fixed by choosing one
of the two orientations.  Let $C$ be a cycle inducing the cycle
$\kappa$ in $\skel(\mu)$.  Then the relative position $\pos_{C}(C')$
of a cycle $C' \ne C$ is determined by $\mu$ if and only if $C'$ is a
cycle in $\skel(\mu)$ or if it is contracted belonging to a virtual
edge not contained in $\kappa$.  Since we consider only one of the two
embeddings of $\skel(\mu)$ at the moment, $\pos_{C}(C')$ is fixed to
one of the two values ``left'' or ``right'' in this case.  The same
can be done for all other cycles that are cycles in $\skel(\mu)$
yielding a fixed value for all relative positions that are determined
by $\mu$.  Finally, we have a partition of all positions determined by
$\mu$ into the set of positions $\mathcal P_1 =
\{\pos_{C_{a(1)}}(C_{b(1)}), \dots, \pos_{C_{a(r)}}(C_{b(r)})\}$ all
having the value ``left'' and the set of positions $\mathcal P_2 =
\{\pos_{C_{c(1)}}(C_{d(1)}), \dots, \pos_{C_{c(s)}}(C_{d(s)})\}$
having the value ``right''.  Now if the embedding of $\skel(\mu)$ is
not fixed anymore, we have only the possibility to flip it.  By
flipping, all the positions in $\mathcal P_1$ change to ``right'' and
all positions in $\mathcal P_2$ change to ``left''.  Hence, we obtain
that the equation $\pos_{C_{a(1)}}(C_{b(1)}) = \dots =
\pos_{C_{a(r)}}(C_{b(r)}) \not= \pos_{C_{c(1)}}(C_{d(1)}) = \dots =
\pos_{C_{c(s)}}(C_{d(s)})$ is satisfied for every embedding of the
cycles $\mathcal C = \{C_1, \dots, C_k\}$ induced by an embedding of
$G$.

To sum up, we obtain a set of (in)equalities relating the relative
positions of cycles to one another.  We call these constraints the
\emph{PR-node constraints} with respect to the biconnected graph~$G$.
Obviously the PR-node constraints are necessary in the sense that
every embedding of $G$ induces an embedding of the cycles $\mathcal C
= \{C_1, \dots, C_k\}$ satisfying these constraints.  The following
lemma additionally states the sufficiency of the PR-node constraints.

\begin{lemma}
  \label{lem:PR-node-constraints}
  Let $G$ be a biconnected planar graph containing the disjoint cycles
  $\mathcal C = \{C_1, \dots, C_k\}$.  Let further $\mathcal
  E_{\mathcal C}$ be a semi-embedding of these cycles.  There is an
  embedding $\mathcal E$ of $G$ inducing $\mathcal E_{\mathcal C}$ if
  and only if $\mathcal E_{\mathcal C}$ satisfies the PR-node
  constraints.
\end{lemma}
\begin{proof}
  The ``only if''-part of the proof is obvious, as mentioned above.
  It remains to show the ``if''-part.  Let $\mathcal E_{\mathcal C}$
  be a semi-embedding of $\mathcal C = \{C_1, \dots, C_k\}$ satisfying
  the PR-node constraints given by~$G$.  We show how to construct an
  embedding $\mathcal E$ of $G$ inducing the embedding $\mathcal
  E_{\mathcal C}$ on the cycles~$\mathcal C = \{C_1, \dots, C_k\}$.
  We simply process the nodes of the SPQR-tree one by one and choose
  an embedding for the skeleton of every node.  Let $\mu$ be a node in
  $\mathcal T$.  If $\mu$ is an S- or a Q-node, there is nothing to
  do, since there is no choice for the embedding of $\skel(\mu)$.  If
  $\mu$ is a P-node several relative positions may be determined by
  the embedding of $\skel(\mu)$.  However, these positions satisfy the
  PR-node constraints stemming from $\mu$, hence we can choose an
  embedding for $\skel(\mu)$ determining these positions as given
  by~$\mathcal E_{\mathcal C}$.  Obviously, the same holds for the
  case where $\mu$ is an R-node.  Hence, we finally obtain an
  embedding~$\mathcal E$ of $G$ determining the positions that are
  determined by a node in $\mathcal T$ as required by $\mathcal
  E_{\mathcal C}$.  Due to Lemma~\ref{lem:determining-pos-by-ex-one}
  every pair of relative positions is determined by exactly one node
  in $\mathcal T$, yielding that the resulting embedding $\mathcal E$
  induces $\mathcal E_{\mathcal C}$ on the cycles.  Note that this
  shows implicitly that $\mathcal E_{\mathcal C}$ is not only a
  semi-embedding but also an embedding.
\end{proof}

Now let $\1G$ and $\2G$ be two biconnected planar graphs with the
common graph $G$ consisting of pairwise disjoint simple cycles
$\mathcal C = \{C_1, \dots, C_k\}$.  If we find a semi-embedding
$\mathcal E$ of the cycles that satisfies the PR-node constraints with
respect to $\1G$ and $\2G$ simultaneously, we can use
Lemma~\ref{lem:PR-node-constraints} to find embeddings $\1{\mathcal
  E}$ and $\2{\mathcal E}$ for $\1G$ and $\2G$ both inducing the
embedding $\mathcal E$ on the common graph~$G$.  Thus, satisfying the
PR-node constraints with respect to both $\1G$ and $\2G$, is
sufficient to find a {\sc SEFE}.  Conversely, given a pair of
embeddings $\1{\mathcal E}$ and $\2{\mathcal E}$ inducing the same
embedding $\mathcal E$ on~$G$, this embedding $\mathcal E$ needs to
satisfy the PR-node constraints with respect to both, $\1G$ and $\2G$,
which is again due to Lemma~\ref{lem:PR-node-constraints}.  Since the
PR-node constraints form a set of boolean (in)equalities we can
express them as an instance of {\sc 2-Sat}.  As this instance has
polynomial size and can easily be computed in polynomial time, we
obtain the following theorem.

\begin{theorem} 
  \label{thm:sefe-biconnected-quadratic}
  {\sc Simultaneous Embedding with Fixed Edges} can be solved in
  quadratic time for biconnected graphs whose common graph is a set of
  disjoint cycles.
\end{theorem}
\begin{proof}
  It remains to show that the PR-node constraints can be computed in
  quadratic time, yielding an instance of {\sc 2-Sat} with quadratic
  size.  As this {\sc 2-Sat} instance can be solved consuming time
  linear in its size~\cite{ComplexityofTimetable-Even.etal(76),
    linear-timealgorithmtesting-Aspvall.etal(79)}, we obtain a
  quadratic-time algorithm.

  We show how to process each node $\mu$ of the SPQR-tree in $\mathcal
  O(n \cdot |\skel(\mu)|)$ time, computing the PR-node constraints
  stemming from $\mu$.  For each virtual edge $\eps$ we compute a list
  of cycles in $\mathcal C$ that contain edges in the expansion graph
  $\expan(\eps)$ by traversing all leaves in the corresponding
  subtree, consuming $\mathcal O(n)$ time for each virtual edge.  Then
  the list of cycles that occur as cycles in $\skel(\mu)$ can be
  computed in linear time.  For each of these cycles $C$ all
  constraints on relative positions with respect to $C$ determined by
  $\mu$ can be easily computed in $\mathcal O(n)$ time.  As only
  $\mathcal O(|\skel(\mu)|)$ cycles can be contained as cycles in
  $\skel(\mu)$, this yields the claimed $\mathcal O(n \cdot
  |\skel(\mu)|)$ time for each skeleton.  Since the total size of the
  skeletons is linear in the size of the graph, this yields an overall
  $\mathcal O(n^2)$-time algorithm.
\end{proof}

\subsubsection*{Allowing Cutvertices}
\label{sec:cutvertices}

In this section we consider the case where the graphs may contain
cutvertices.  As before, we consider a single graph $G$ containing a
set of disjoint cycles $\mathcal C = \{C_1, \dots, C_k\}$ first.  Let
$C \in \mathcal C$ be one of the cycles and let $v$ be a cutvertex
contained in the same block $B$ that contains~$C$.  The cutvertex $v$
splits $G$ into $\ell$ cut components $H_1, \dots, H_\ell$.  Assume
without loss of generality that $B$ (and with it also $C$) is
contained in $H_1$.  We distinguish between the cases that $v$ is
contained in $C$ and that it is not.

If {\bf $\boldsymbol{v}$ is not contained in $\boldsymbol{C}$}, then
the relative position $\pos_{C}(v)$ is determined by the embedding of
the block $B$ and it follows that all the subgraphs $H_2, \dots,
H_\ell$ lie on the same side of $C$ as $v$ does.  It follows from the
biconnected case (Lemma~\ref{lem:determining-pos-by-ex-one}) that
$\pos_{C}(v)$ is determined by the embedding of the skeleton of
exactly one node $\mu$ in the SPQR-tree of $B$.  Obviously, the
conditions that all cycles in $H_2, \dots, H_\ell$ are on the same
side of $C$ as $v$ can be easily added to the PR-node constraints
stemming from the node $\mu$; call the resulting constraints the
\emph{extended PR-node constraints}.  These constraints are clearly
necessary.  On the other hand, if $\mathcal E_{\mathcal C}$ is a
semi-embedding of the cycles satisfying the extended PR-node
constraints, we can find an embedding $\mathcal E_B$ of the block $B$
such that all relative positions of cycles that are determined by
single nodes in the SPQR-tree of $B$ are compatible with~$\mathcal
E_{\mathcal C}$.

If {\bf $\boldsymbol{v}$ is contained in $\boldsymbol{C}$}, the
relative position $\pos_{C}(v)$ does not exist.  Assume the embedding
of each block is already chosen.  Then for each subgraph $H \in \{H_2,
\dots, H_\ell\}$, the positions $\pos_{C}(H)$ can be chosen
arbitrarily and independently.  In this case we say for a cycle $C'$
in $H$ that its relative position $\pos_{C}(C')$ is \emph{determined
  by the embedding chosen for the cutvertex $v$}.  Obviously, in every
embedding of $G$, a pair of cycles $C_i$ and $C_j$ both belonging to
the same subgraph $H \in \{H_2, \dots, H_\ell\}$ lie on the same side
of $C$ yielding the equation $\pos_{C}(C_i) = \pos_{C}(C_j)$.  This
equation can be set up for every pair of cycles in each of the
subgraphs, yielding the \emph{cutvertex constraints} with respect
to~$v$.  Again we have that, given a semi-embedding $\mathcal
E_{\mathcal C}$ of the cycles satisfying the cutvertex constraints
with respect to $v$, we can simply choose an embedding of the graph
such that the relative positions determined by the embedding around
the cutvertex are compatible with~$\mathcal E_{\mathcal C}$.

To sum up, a semi-embedding $\mathcal E_{\mathcal C}$ on the cycles
$\mathcal C = \{C_1, \dots, C_k\}$ that is induced by an
embedding~$\mathcal E$ of the whole graph always satisfies the
extended PR-node and cutvertex constraints.  Moreover, given a
semi-embedding $\mathcal E_{\mathcal C}$ satisfying these constraints,
we can find an embedding $\mathcal E$ of $G$ inducing compatible
relative positions for each relative position that is determined by a
single node in the SPQR-tree of a block or by a cutvertex.  Obviously,
the relative position of every pair of cycles is determined by such a
node or a cutvertex.  Thus the extended PR-node and cutvertex
constraints together are sufficient, that is, given a semi-embedding
of the cycles satisfying these constraints, we can find an embedding
of $G$ inducing this semi-embedding.  This shows implicitly that the
given semi-embedding is an embedding.  This result is stated again in
the following lemma.

\begin{lemma}
  \label{lem:ext-PR-node-cutvert-constraints}
  Let $G$ be a connected planar graph containing the disjoint cycles
  $\mathcal C = \{C_1, \dots, C_k \}$.  Let further~$\mathcal E_{\mathcal
    C}$ be a semi-embedding of these cycles.  There is an embedding
  $\mathcal E$ of $G$ inducing~$\mathcal E_{\mathcal C}$ if and only
  if $\mathcal E_{\mathcal C}$ satisfies the extended PR-node and
  cutvertex constraints with respect to~$G$.
\end{lemma}

This result again directly yields a polynomial-time algorithm to solve
{\sc SEFE} for the case that both graphs $\1G$ and $\2G$ are connected
and their common graph $G$ consists of a set of disjoint cycles.
Moreover, requiring both graphs to be connected is not really a
restriction due to Theorem~\ref{thm:connecting-graphs}.  The extended
PR-node and cutvertex constraints can be computed similarly as in the
proof of Theorem~\ref{thm:sefe-biconnected-quadratic}, yielding the
following theorem.

\begin{theorem}
  {\sc Simultaneous Embedding with Fixed Edges} can be solved in
  quadratic time if the common graph consists of disjoint cycles.
\end{theorem}

Note that we really need to use Theorem~\ref{thm:connecting-graphs} to
ensure that the graphs are connected since our approach does not
extend to the case where the graphs are allowed to be disconnected.
In this case it would still be easy to formulate necessary conditions
in terms of boolean equations.  However, these conditions would only
be sufficient if it is additionally ensured that the given
semi-embedding actually is an embedding.  The reason why this is not
directly ensured by the embedding of the graph (as it is in the
connected case) is that the relative position of cycles to one another
is not determined by exactly one choice that can be made independently
from the other choices; see Figure~\ref{fig:disonnected-graphs}.

\begin{figure}
  \centering
  \includegraphics{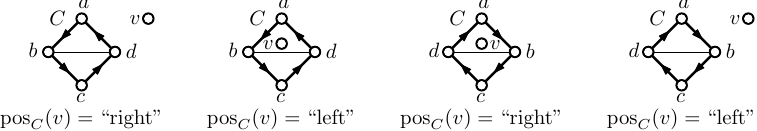}
  \caption{One component containing $C$ (bold) and another consisting
    only of the vertex $v$.  Changing the face in which $v$ lies may
    change the relative position $\pos_C(v)$.  Moreover, changing the
    embedding of the component containing $C$ (in this case flipping
    it) also changes $\pos_C(v)$.}
  \label{fig:disonnected-graphs}
\end{figure}

\subsection{A Compact Representation of all Simultaneous Embeddings}
\label{sec:compact-rep}

In the previous section we showed that {\sc SEFE} can be solved in
polynomial time for the case that the common graph consists of
disjoint cycles.  In this section we describe a data structure, the
\emph{CC-tree}, representing all embeddings of a set of disjoint cycles that
can be induced by an embedding of a connected graph containing them.
Afterwards, we show that the intersection of the sets of embeddings
represented by two CC-trees can again be represented by a CC-tree.  In
Section~\ref{sec:line-time-algor} we then show that the CC-tree and
the intersection of two CC-trees can be computed in linear time,
yielding an optimal linear-time algorithm for {\sc SEFE} for the case
that the common graph consists of disjoint cycles.  Note that this
algorithm obviously extends to the case where $k$ graphs $\1G, \dots,
\k G$ are given such that they all intersect in the same common graph
$G$ consisting of a set of disjoint cycles.

\subsubsection*{C-Trees and CC-Trees}
\label{sec:cc-trees}

Let $\mathcal C = \{C_1, \dots, C_k\}$ be a set of disjoint cycles.  A
\emph{cycle-tree (C-tree)} $\mathcal T_{\mathcal C}$ on these cycles
is a minimal connected graph containing $\mathcal C$; see
Figure~\ref{fig:cc-tree}.  Obviously, every embedding of $\mathcal
T_{\mathcal C}$ induces an embedding of the cycles.  We say that two
embeddings of $\mathcal T_{\mathcal C}$ are equivalent if they induce
the same embedding of $\mathcal C$ and we are only interested in the
equivalence classes with respect to this equivalence relation.  An
embedding $\mathcal E$ of the cycles in $\mathcal C$ is
\emph{represented} by $\mathcal T_{\mathcal C}$ if it admits an
embedding inducing $\mathcal E$.  Note that contracting each of the
cycles $\mathcal C = \{C_1, \dots, C_k\}$ in a C-tree to a single
vertex yields a spanning tree on these vertices.  In most cases we
implicitly assume the cycles to be contracted such that $\mathcal
T_{\mathcal C}$ can be treated like a tree.

\begin{figure}
  \centering
  \includegraphics[page=1]{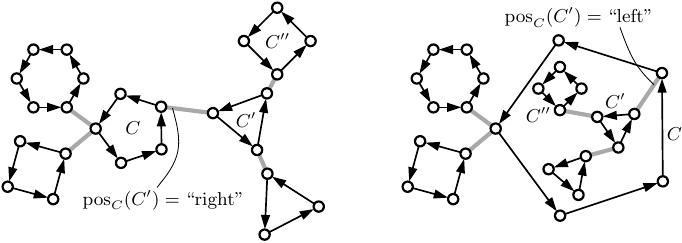}
  \caption{Two embeddings of the same CC-tree.  The only difference
    between the embeddings is that different values are chosen for the
    crucial relative position $\pos_C(C')$.  Note that the tree
    structure enforces the (non-crucial) relative position
    $\pos_C(C'')$ to be equal to $\pos_C(C')$.}
  \label{fig:cc-tree}
\end{figure}

The embedding choices that can be made for $\mathcal T_{\mathcal C}$
are of the following kind.  For every edge $e = \{C, C'\}$ in
$\mathcal T_{\mathcal C}$, we can decide to put all cycles in the
subtree attached to $C$ via~$e$ either to the left or to the right of
$C$.  In particular, we can assign a value ``left'' or ``right'' to
the relative position $\pos_C(C')$.  Moreover, by fixing the relative
positions $\pos_C(C')$ and $\pos_{C'}(C)$ for every pair of cycles $C$
and $C'$ that are adjacent in $\mathcal T_{\mathcal C}$, the embedding
represented by $\mathcal T_{\mathcal C}$ is completely determined.
Thus, given a C-tree $\mathcal T_{\mathcal C}$, we call the relative
positions $\pos_C(C')$ and $\pos_{C'}(C)$ with $C, C' \in \mathcal C$
\emph{crucial} if $C$ and~$C'$ are adjacent in $\mathcal T_{\mathcal
  C}$; see Figure~\ref{fig:cc-tree}.  We note that, when determining
an embedding of~$\mathcal T_{\mathcal C}$, the crucial relative
positions can be chosen independently from one another.

Since the crucial relative positions with respect to a C-tree
$\mathcal T_{\mathcal C}$ are binary variables, we can use
(in)equalities between them to further constrain the embeddings
represented by $\mathcal T_{\mathcal C}$.  We call a C-tree with such
additional constraints on its crucial relative positions a
\emph{constrained cycle-tree (CC-tree)} on the set of cycles $\mathcal
C$.  In this way, there is a bijection between the embeddings
of~$\mathcal C$ represented by a CC-tree and the solutions of an
instance of {\sc 2-Sat} given by the constraints on the crucial
relative positions of~$\mathcal T_{\mathcal C}$.  We essentially prove
two things.  First, for every connected graph $G$ containing the
cycles $\mathcal C$, there exists a CC-tree representing exactly the
embeddings of $\mathcal C$ that can be induced by embeddings of $G$.
Essentially, we have to restrict the extended PR-node and cutvertex
constraints to the crucial relative positions of a C-tree compatible
with $G$.  Second, for a pair of CC-trees~$\1{\mathcal T_{\mathcal
    C}}$ and $\2{\mathcal T_{\mathcal C}}$ on the same set~$\mathcal
C$ of cycles, there exists a CC-tree $\mathcal T_{\mathcal C}$
representing exactly the embeddings of $\mathcal C$ that are
represented by $\1{\mathcal T_{\mathcal C}}$ and $\2{\mathcal
  T_{\mathcal C}}$.

Let~$G$ be a connected planar graph containing a set~$\mathcal C$ of
disjoint cycles.  We say that a C-tree $\mathcal T_{\mathcal C}$ is
\emph{compatible} with $G$ if it is a minor of $G$, that is if it can
be obtained by contracting edges in a subgraph of $G$.  The
corresponding \emph{compatible CC-tree} is obtained from~$\mathcal
T_{\mathcal C}$ by adding the subset of the extended PR-node and
cutvertex constraints that only involve crucial relative positions
of~$\mathcal T_{\mathcal C}$.  Note that there may be many compatible
CC-trees for a single graph $G$.  However, in the following we
arbitrarily fix one of them and speak about \emph{the} CC-tree of~$G$.

\begin{theorem}
  \label{thm:cycle-tree-is-rep}
  Let $G$ be a connected planar graph containing the disjoint cycles
  $\mathcal C = \{C_1,\dots,C_k\}$.  The CC-tree $\mathcal T_{\mathcal C}$ of $G$
  represents exactly the embeddings of $\mathcal C$ that can be
  induced by an embedding of $G$.
\end{theorem}
\begin{proof}
  Let $\mathcal E$ be an embedding of $G$ and let $\mathcal
  E_{\mathcal C}$ be the embedding induced on the cycles $\mathcal C =
  \{C_1, \dots, C_k\}$.  Obviously, the CC-tree $\mathcal T_{\mathcal
    C}$ can be obtained from $G$ by contracting the cycles $\mathcal
  C$ to single vertices, choosing a spanning tree, expanding the
  cycles and contracting edges incident to non-cycle vertices.  Since
  we essentially only pick a subgraph of $G$ containing all cycles
  $\mathcal C = \{C_1, \dots, C_k\}$ and contract edges, the embedding
  $\mathcal E_{\mathcal C}$ is preserved.  Moreover, by
  Lemma~\ref{lem:ext-PR-node-cutvert-constraints}, it satisfies the
  extended PR-node and cutvertex constraints since it is induced by
  the embedding $\mathcal E$ of $G$.  Hence, $\mathcal E_{\mathcal C}$
  is represented by the CC-tree~$\mathcal T_{\mathcal C}$.

  Conversely, let $\mathcal E_{\mathcal C}$ be an embedding on the
  cycles represented by the CC-tree $\mathcal T_{\mathcal C}$.  By
  definition, the extended PR-node and cutvertex constraints are
  satisfied for the crucial relative positions.  We show that the
  tree-like structure of $\mathcal T_{\mathcal C}$ ensures that they
  are also satisfied for the remaining relative positions, yielding
  that an embedding $\mathcal E$ of $G$ inducing $\mathcal E_{\mathcal
    C}$ exists due to Lemma~\ref{lem:ext-PR-node-cutvert-constraints}.
  We start with the PR-node constraints.  Let $B$ be a block of $G$
  with SPQR-tree $\mathcal T(B)$.  In a P-node $\mu$ containing a
  cycle $C$ as cycle $\kappa$ every other cycle in $B$ is contracted,
  belonging to a single virtual edge.  Let~$C_i$ and $C_j$ be two
  cycles in $B$ belonging to the same virtual edge $\eps$ not
  contained in $\kappa$.  In this case the PR-node constraints
  stemming from $\mu$ require $\pos_C(C_i) = \pos_C(C_j)$, and we show
  that this equation is implied if the extended PR-node constraints
  are satisfied for the crucial relative positions.  Let~$C_i'$ and
  $C_j'$ be the first cycles on the paths from $C$ to $C_i$ and $C_j$
  in $\mathcal T_{\mathcal C}$, respectively.  Note that~$C_i'$
  and~$C_j'$ are not necessarily contained in the block $B$.  However,
  we first consider the case where both are contained in $B$.  Then
  $C_i'$ and $C_j'$ are both contracted in the same virtual
  edge~$\eps$ as $C_i$ and $C_j$ since a path from a cycle belonging
  to $\eps$ to a cycle belonging to a different virtual edge would
  necessarily contain a pole of $\skel(\mu)$ and thus a vertex in $C$.
  Thus, the PR-node constraints restricted to the crucial relative
  positions enforce $\pos_{C}(C_i') = \pos_{C}(C_j')$.  Furthermore,
  the tree structure of $\mathcal T_{\mathcal C}$ enforces
  $\pos_{C}(C_i) = \pos_{C}(C_i')$ and $\pos_{C}(C_j) =
  \pos_{C}(C_j')$.  Hence, in this case the PR-node constraints
  stemming from $\mu$ are implied by their restriction to the crucial
  relative positions.  For the case that $C_i'$ or $C_j'$ are contained
  in a different block, they are connected to $B$ via
  cutvertices~$v_i$ or $v_j$, which must belong to $\expan(\eps)$ by the
  same argument as above, namely that every path from $C_i$ or $C_j$ to a
  vertex that is contained in the expansion graph of another virtual
  edge needs to contain one of the poles.  Thus, the extended PR-node
  constraints enforce $\pos_{C}(C_i') = \pos_{C}(C_j')$ yielding the
  same situation as above.  In total, the extended PR-node constraints
  stemming from a P-node $\mu$ restricted to the crucial relative
  positions enforce that the PR-node constraints stemming from $\mu$
  are satisfied for all relative positions.

  For the case that $\mu$ is an R-node a similar argument holds.  If
  $C$ is a cycle $\kappa$ in $\skel(\mu)$ and two cycles $C_i$ and
  $C_j$ lie contracted or as cycles on the same side (on different
  sides) of $\kappa$, then the first cycles $C_i'$ and $C_j'$ on the
  path from $C$ to $C_i$ and $C_j$ in the CC-tree $\mathcal
  T_{\mathcal C}$ lie on the same side (on different sides) of
  $\kappa$ or the cutvertices connecting $C_i'$ and $C_j'$ to the
  block $B$ lie on the same side (on different sides) of $\kappa$.
  Thus, the extended PR-node constraints restricted to the crucial
  relative positions enforce $\pos_C(C_i') = \pos_C(C_j')$
  ($\pos_C(C_i') \not= \pos_C(C_j')$) and the tree structure of
  $\mathcal T_{\mathcal C}$ yields $\pos_C(C_i) = \pos_C(C_i')$ and
  $\pos_C(C_j) = \pos_C(C_j')$.  Obviously, these arguments extend to
  the case of extended PR-node constraints since a cutvertex not
  contained in $C$ can be treated like a disjoint cycle.

  It remains to deal with the cutvertex constraints stemming from the
  case where $C$ is a cycle containing a cutvertex $v$ splitting $G$
  into the cut components $H_1, \dots, H_\ell$.  Let without loss of
  generality $H_1$ be the subgraph containing $C$.  The cutvertex
  constraints ensure that a pair of cycles $C_i$ and $C_j$ belonging
  to the same subgraph $H \in \{H_2, \dots, H_\ell\}$ are located on
  the same side of $C$.  Let $C_i'$ and~$C_j'$ be the first cycles on
  the path from $C$ to $C_i$ and~$C_j$ in the CC-tree $\mathcal
  T_{\mathcal C}$, respectively.  Obviously~$C_i'$ and~$C_j'$ belong
  to the same subgraph $H$ and hence the cutvertex constraints
  restricted to the crucial relative positions enforce $\pos_C(C_i') =
  \pos_C(C_j')$.  Moreover, the tree structure of $\mathcal
  T_{\mathcal C}$ again ensures that the equations $\pos_C(C_i) =
  \pos_C(C_i')$ and $\pos_C(C_j) = \pos_C(C_j')$ hold, which concludes
  the proof.
\end{proof}

\subsubsection*{Intersecting CC-Trees}
\label{sec:inters-cc-trees}

In this section we consider two CC-trees $\1{\mathcal T_{\mathcal C}}$
and $\2{\mathcal T_{\mathcal C}}$ on the same set of cycles $\mathcal
C$.  We show that the set of embeddings that are represented by both
$\1{\mathcal T_{\mathcal C}}$ and $\2{\mathcal T_{\mathcal C}}$ can
again be represented by a single CC-tree.  We will show this by
constructing a new CC-tree, which we call the \emph{intersection}
of~$\1{\mathcal T_{\mathcal C}}$ and $\2{\mathcal T_{\mathcal C}}$,
showing afterwards that this CC-tree has the desired property.  The
intersection $\mathcal T_{\mathcal C}$ is a copy of $\1{\mathcal
  T_{\mathcal C}}$ with some additional constraints given by the
second CC-tree $\2{\mathcal T_{\mathcal C}}$.  We essentially have to
formulate two types of constraints.  First, constraints stemming from
the structure of the underlying C-tree of $\2{\mathcal T_{\mathcal
    C}}$.  Second, the constraints given by the (in)equalities on the
relative positions that are crucial with respect to $\2{\mathcal
  T_{\mathcal C}}$.  We show that both kinds of constraints can be
formulated as (in)equalities on the relative positions that are
crucial with respect to $\1{\mathcal T_{\mathcal C}}$.

Let $C_1$ and $C_2$ be two cycles joined by an edge in $\2{\mathcal
  T_{\mathcal C}}$.  Obviously, $C_1$ and $C_2$ are contained in the
boundary of a common face in every embedding $\2{\mathcal E}$
represented by $\2{\mathcal T_{\mathcal C}}$.  It is easy to formulate
constraints on the relative positions that are crucial with respect to
$\1{\mathcal T_{\mathcal C}}$ such that~$C_1$ and~$C_2$ are contained
in the boundary of a common face for every embedding represented by
$\1{\mathcal T_{\mathcal C}}$.  Consider the path $\pi$ from $C_1$ to
$C_2$ in $\1{\mathcal T_{\mathcal C}}$.  For every three cycles $C$,
$C'$ and $C''$ appearing consecutively on $\pi$ it is necessary that
$\pos_{C'}(C) = \pos_{C'}(C'')$ holds.  Otherwise $C_1$ and $C_2$ would
be separated by $C'$.  Conversely, if this equation holds for every
triple of consecutive cycles on $\pi$, then $C_1$ and $C_2$ always lie
on a common face.  We call the resulting equations the
\emph{common-face constraints}.  Note that all relative positions
involved in such constraints are crucial with respect to $\1{\mathcal
  T_{\mathcal C}}$.

To formulate the constraints given on the crucial relative positions
of $\2{\mathcal T_{\mathcal C}}$, we essentially find, for each of
these crucial relative positions $\pos_{C_1}(C_2)$, a relative
position $\pos_{C_1}(C_2')$ that is crucial with respect to
$\1{\mathcal T_{\mathcal C}}$ such that $\pos_{C_1}(C_2)$ is
determined by fixing $\pos_{C_1}(C_2')$ in $\1{\mathcal T_{\mathcal
    C}}$.  More precisely, for every relative position
$\pos_{C_1}(C_2)$ that is crucial with respect to $\2{\mathcal
  T_{\mathcal C}}$ we define its \emph{representative} in $\1{\mathcal
  T_{\mathcal C}}$ to be the crucial relative position
$\pos_{C_1}(C_2')$, where $C_2'$ is the first cycle in $\1{\mathcal
  T_{\mathcal C}}$ on the path from $C_1$ to $C_2$.  We obtain the
\emph{crucial-position constraints} on the crucial relative positions
of $\1{\mathcal T_{\mathcal C}}$ by replacing every relative position
in the constraints given for $\2{\mathcal T_{\mathcal C}}$ by its
representative.  The resulting set of (in)equalities on the crucial
relative positions of $\1{\mathcal T_{\mathcal C}}$ is obviously
necessary.

We can now formally define the \emph{intersection} $\mathcal
T_{\mathcal C}$ of two CC-trees $\1{\mathcal T_{\mathcal C}}$ and
$\2{\mathcal T_{\mathcal C}}$ to be $\1{\mathcal T_{\mathcal C}}$ with
the common-face and crucial-position constraints additionally
restricting its crucial relative positions.  We obtain the following
theorem, justifying the name ``intersection''.

\begin{theorem}
  \label{thm:intersection-is-intersection}
  The intersection of two CC-trees represents exactly the embeddings
  that are represented by both CC-trees.
\end{theorem}
\begin{proof}
  Let $\1{\mathcal T_{\mathcal C}}$ and $\2{\mathcal T_{\mathcal C}}$
  be two CC-trees and let $\mathcal T_{\mathcal C}$ be their
  intersection.  Let further $\mathcal E$ be an embedding represented
  by $\1{\mathcal T_{\mathcal C}}$ and $\2{\mathcal T_{\mathcal C}}$.
  Then $\mathcal T_{\mathcal C}$ also represents $\mathcal E$ since
  the common-face and crucial-position constraints are obviously
  necessary.  Now let $\mathcal E$ be an embedding represented
  by~$\mathcal T_{\mathcal C}$.  It is clearly also represented by
  $\1{\mathcal T_{\mathcal C}}$ since $\mathcal T_{\mathcal C}$ is the
  same tree with some additional constraints.  It remains to show that
  $\mathcal E$ is represented by $\2{\mathcal T_{\mathcal C}}$.  The
  embedding $\mathcal E$ induces a value for every relative position.
  In particular, it induces a value for every relative position that
  is crucial with respect to~$\2{\mathcal T_{\mathcal C}}$.  The
  crucial-position constraints ensure that these values satisfy the
  constraints given for the crucial relative positions in the CC-tree
  $\2{\mathcal T_{\mathcal C}}$.  Thus we can simply take these
  positions, apply them to $\2{\mathcal T_{\mathcal C}}$ and obtain an
  embedding $\2{\mathcal E}$ that is represented by $\2{\mathcal
    T_{\mathcal C}}$.  It remains to show that~$\mathcal E =
  \2{\mathcal E}$.  To this end, we consider an arbitrary pair of
  cycles $C_1$ and $C_2$ and show the following equation, where
  $\pos_{C_1}(C_2)$ and $\2{\pos_{C_1}}(C_2)$ denote the relative
  positions of $C_2$ with respect to $C_1$ in the embeddings~$\mathcal
  E$ and $\2{\mathcal E}$, respectively.
  \begin{eqnarray}
    \label{eq:1}
    \pos_{C_1}(C_2)&=&\2{\pos_{C_1}}(C_2)
  \end{eqnarray}

  Consider the paths $\pi$ and $\2\pi$ from $C_1$ to $C_2$ in
  $\mathcal T_{\mathcal C}$ and $\2{\mathcal T_{\mathcal C}}$,
  respectively.  We use induction on the length of $\2\pi$,
  illustrated in Figure~\ref{fig:intersection-CC-trees}, with
  Equation~\eqref{eq:1} as induction hypothesis.  If $|\2\pi| = 1$,
  then $\pos_{C_1}(C_2)$ is crucial with respect to $\2{\mathcal
    T_{\mathcal C}}$ and thus equal in both embeddings $\mathcal E$
  and $\2{\mathcal E}$ by construction of $\2{\mathcal E}$.  For the
  case $|\2\pi| > 1$ let $C_1'$ and $\2{C_1}$ be the neighbors of
  $C_1$ in $\pi$ and $\2\pi$, respectively.  Since $C_1'$ and
  $\2{C_1}$ lie on the path between $C_1$ and $C_2$ in $\mathcal
  T_{\mathcal C}$ and $\2{\mathcal T_{\mathcal C}}$, the following two
  equations hold.
  \begin{eqnarray}
    \label{eq:2}
    \pos_{C_1}(C_2)&=&\pos_{C_1}(C_1') \\
    \label{eq:3}
    \2{\pos_{C_1}}(C_2)&=&\2{\pos_{C_1}}(\2{C_1})
  \end{eqnarray}
  Thus, it suffices to show that $\pos_{C_1}(C_1') =
  \2{\pos_{C_1}}(\2{C_1})$ holds to obtain Equation~\eqref{eq:1}.  Let
  $\2{C_2}$ be the neighbor of $C_2$ on the path $\2\pi$.  Since the
  path from $C_1$ to $\2{C_2}$ is shorter than $\2\pi$ the equation
  $\pos_{C_1}(\2{C_2}) = \2{\pos_{C_1}}(\2{C_2})$ follows from the
  induction hypothesis stated in Equation~\eqref{eq:1}.  There are two
  possibilities.  The path from $C_1$ to $\2{C_2}$ in $\mathcal
  T_{\mathcal C}$ has either $\{C_1, C_1'\}$ or $\{C_1, C_1''\}$ for
  some other cycle $C_1''$ as first edge.  In the former case the
  equation
  \begin{eqnarray}
    \label{eq:4}
    \pos_{C_1}(C_1')&=&\2{\pos_{C_1}}(\2{C_1})
  \end{eqnarray}
  obviously follows.  Together with Equations~\eqref{eq:2}
  and~\eqref{eq:3}, this yields the induction hypothesis
  (Equation~\eqref{eq:1}).  In the latter case we have the following
  equation.
  \begin{eqnarray}
    \label{eq:5}
    \pos_{C_1}(C_1'')&=&\2{\pos_{C_1}}(\2{C_1})
  \end{eqnarray}
  Moreover, the common-face constraints stemming from the edge
  $\{\2{C_2}, C_2\}$ in $\2{\mathcal T_{\mathcal C}}$ enforce 
  \begin{eqnarray}
    \label{eq:6}
    \pos_{C_1}(C_1'')&=&\pos_{C_1}(C_1'),
  \end{eqnarray}
  again yielding the induction hypothesis stated in
  Equation~\eqref{eq:1}.  This concludes the proof.
\end{proof}

\begin{figure}
  \centering
  \includegraphics[page=1]{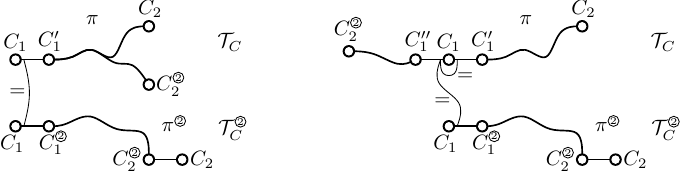}
  \caption{The two cases arising in the proof of
    Theorem~\ref{thm:intersection-is-intersection}.  If the path from
    $C_1$ to $\2{C_2}$ starts with the edge $\{C_1, C_1'\}$ (left) the
    equation $\pos_{C_1}(C_1') = \2{\pos_{C_1}}(\2{C_1})$ follows by
    induction.  Otherwise (right) $\pos_{C_1}(C_1'') =
    \2{\pos_{C_1}}(\2{C_1})$ follows by induction and the equation
    $\pos_{C_1}(C_1') = \pos_{C_1}(C_1'')$ holds due to the
    common-face constraint stemming from $\{C_2, \2{C_2}\}$.}
  \label{fig:intersection-CC-trees}
\end{figure}

\subsection{Linear-Time Algorithm}
\label{sec:line-time-algor}

In this section we first show how to compute the CC-tree of a given
graph containing a set of disjoint cycles in linear time.  Afterwards,
we show that the intersection of two CC-trees can be computed in
linear time.  Together, this yields a linear-time algorithm for the
variant of {\sc SEFE} we consider.

\subsubsection*{Computing the CC-Tree in Linear Time}
\label{sec:computing-cc-tree}

\todo{Pag 15, line 1 of ``Computing...time'':}The first step of
computing the CC-tree $\mathcal T_{\mathcal C}$ of a graph $G$ is to
compute the underlying C-tree.  Obviously, this can be easily done in
linear time.  Thus, the focus of this section lies on computing the
extended PR-node and cutvertex constraints restricted to the crucial
relative positions.  To simplify notation we first consider the case
where $G$ is biconnected.  Before we start computing the PR-node
constraints we need one more definition.  For each cycle $C$ there is
a set of inner nodes in the SPQR-tree $\mathcal T$ containing $C$ as a
cycle.  We denote the subgraph of $\mathcal T$ induced by these nodes
by~$\proj{\mathcal T}{C}$ and call it the \emph{induced subtree} with
respect to $C$.  To justify the term ``subtree'' we prove the
following lemma.

\begin{lemma}
  Let $G$ be a biconnected planar graph with SPQR-tree $\mathcal T$
  containing the disjoint cycles $\mathcal C = \{C_1, \dots, C_k\}$.
  The induced subtrees $\proj{\mathcal T}{C_1}, \dots, \proj{\mathcal
    T}{C_k}$ with respect to $C_1, \dots, C_k$ are pairwise
  edge-disjoint trees.
\end{lemma}
\begin{proof}
  We first show that the induced tree with respect to a single cycle
  is really a tree.  Afterwards, we show that two disjoint cycles
  induce edge-disjoint trees, yielding that they have linear size in
  total.

  Let $C$ be a cycle in $G$ and let $\proj{\mathcal T}{C}$ be its
  induced tree.  A Q-node in $\mathcal T$ contains $C$ as a cycle if
  and only if the corresponding edge is contained in $C$.  For each
  pair of these Q-nodes all nodes on the path between them are
  contained in $\proj{\mathcal T}{C}$, thus the Q-nodes are in the
  same connected component in the induced subtree.  Moreover, an
  internal node in $\mathcal T$ cannot be a leaf in $\proj{\mathcal
    T}{C}$, implying that it contains only one connected component.

  Assume there are two cycles $C_i$ and $C_j$ inducing trees
  $\proj{\mathcal T}{C_i}$ and $\proj{\mathcal T}{C_j}$ that are not
  edge-disjoint.  Let $\{\mu, \mu'\}$ be an edge in $\mathcal T$
  belonging to both.  Let further $\kappa_i$ and $\kappa_j$ be the
  cycles in $\skel(\mu)$ induced by $C_i$ and $C_j$, respectively.
  Since the neighbor $\mu'$ of $\mu$ also contains $C_i$ as a cycle,
  it corresponds to a virtual edge $\eps$ in $\mu$ that is contained
  in $\kappa_i$.  Similarly, $\eps$ is also contained in $\kappa_j$,
  which is a contradiction since $C_i$ and $C_j$ are disjoint.
\end{proof}

Our algorithm computing the PR-node constraints consists of four
phases, each of them consuming linear time.  In each phase we compute
data we then use in the next phase.  Table~\ref{tab:data} gives an
overview about the data we compute.  During all phases we assume the
SPQR-tree $\mathcal T$ to be rooted at a Q-node corresponding to an
edge in $G$ that is not contained in any cycle in $\mathcal C$.  In
the first phase we essentially compute the induced trees
$\proj{\mathcal T}{C}$.  More precisely, for every node $\mu$ in the
SPQR-tree we compute a list $\cyc(\mu)$ containing a cycle $C$ if and
only if $C$ is a cycle in \todo{Pag 16, line 19:}$\skel(\mu)$, that is
if and only if $\mu$ is contained in $\proj{\mathcal T}{C}$.
Moreover, we say a virtual edge $\eps$ in $\skel(\mu)$ \emph{belongs}
to a cycle $C$ if $C$ induces a cycle in $\skel(\mu)$ containing
$\eps$.  Note that $\eps$ belongs to at most one cycle.  If $\eps$
belongs to $C$, we set $\bel(\eps) = C$; if $\eps$ does not belong to
any cycle, we set $\bel(\eps) = \bot$.  Finally, the root of an
induced tree $\proj{\mathcal T}{C}$ with respect to the root chosen
for $\mathcal T$ is denoted by $\roo(\proj{\mathcal T}{C})$.  To sum
up, in the first phase we compute $\cyc(\mu)$ for every node $\mu$,
$\bel(\eps)$ for every virtual edge $\eps$ and $\roo(\proj{\mathcal
  T}{C})$ for every induced subtree $\proj{\mathcal T}{C}$.  In the
second phase, we compute $\high(\mu)$ as the highest edge in the
SPQR-tree $\mathcal T$ on the path from $\mu$ to the root whose
endpoints are both reachable from~$\mu$ without using edges contained
in any of the induced subtrees $\proj{\mathcal T}{C}$.  Note that
$\high(\mu)$ is the edge in~$\mathcal T$ incident to the root if no
edge on the path from $\mu$ to the root is contained in one of the
induced subtrees.  For the special case that the edge from $\mu$ to
its parent itself is already contained in one of the induced trees,
the edge $\high(\mu)$ is not defined and we set $\high(\mu) = \bot$.
In the third phase we compute for every crucial relative position
$\pos_C(C')$ the node in the SPQR-tree determining it, denoted by
$\det(\pos_C(C'))$.  Moreover, for every virtual edge $\eps$ in
$\skel(\mu)$ we compute a list $\contr(\eps)$ of relative positions.
A relative position $\pos_C(C')$ is contained in $\contr(\eps)$ if and
only if it is crucial, determined by $\mu$ and $C'$ is contracted in
$\eps$.  Similarly, the list $\detcyc(\mu)$ for an R-node~$\mu$
contains the crucial relative position $\pos_C(C')$ if and only if $C$
and $C'$ are both cycles in $\skel(\mu)$, implying that $\pos_C(C')$
is determined by $\mu$.  Finally, in the fourth pase, we compute the
PR-node constraints restricted to the crucial relative positions.  The
next lemma states that the first phase can be implemented in linear
time.

\begin{table}
  \heavyrulewidth=1pt
  \centering
  \begin{tabular}{lp{0.8\linewidth}}
    \toprule
    {\bf Data} & {\bf Description}  \\
    \midrule
    $\cyc(\mu)$ & For a node $\mu$ in the SPQR-tree the list of cycles
    in $\mathcal C$ that are cycles in $\skel(\mu)$.  \\
    \midrule
    $\bel(\eps)$ & For a virtual edge $\eps$ in $\skel(\mu)$ either a
    cycle $C \in \mathcal C$ if $C$ induces a cycle in $\skel(\mu)$
    containing $\eps$ or $\bot$ denoting that $\eps$ is not contained
    in such a cycle.  \\
    \midrule
    $\roo(\proj{\mathcal T}{C})$ & The root for the induced tree
    $\proj{\mathcal T}{C}$ with respect to a chosen root for the
    SPQR-tree $\mathcal T$.  \\
    \midrule
    $\high(\mu)$  &  For a node $\mu$ in the SPQR-tree $\mathcal T$
    the highest edge in $\mathcal T$ on the path from $\mu$ to the
    root that is reachable without using an edge in any of the induced
    subtrees~$\proj{\mathcal T}{C}$.  \\
    \midrule
    $\det(\pos_{C}(C'))$ & The node in the SPQR-tree determining the
    relative position $\pos_{C}(C')$ of the cycle $C'$ with respect to
    another cycle $C$.  \\
    \midrule
    $\contr(\eps)$ & For a virtual edge $\eps$ in $\skel(\mu)$ a list
    of relative positions containing $\pos_C(C')$ if and only if it is
    crucial, determined by $\mu$ and $C'$ is contracted in $\eps$.  \\
    \midrule
    $\detcyc(\mu)$ & For every R-node $\mu$ a list of crucial relative
    positions containing $\pos_C(C')$ if and only if $C$ and $C'$ are
    cycles in $\skel(\mu)$. \\
    \bottomrule
\end{tabular}
\caption{Data that is computed to compute the PR-node constraints
  restricted to the crucial relative positions.}
\label{tab:data}
\end{table}

\begin{lemma}
  \label{lem:induced-trees-comp-data}
  Let $G$ be a biconnected planar graph with SPQR-tree~$\mathcal T$
  containing the disjoint cycles~$\mathcal C$.  The data $\cyc(\mu)$
  for every node $\mu$, $\bel(\eps)$ for every virtual edge $\eps$,
  and $\roo(\proj{\mathcal T}{C_i})$ for every cycle~$C_i$ can be
  computed in overall linear time.
\end{lemma}
\begin{proof}
  We process the SPQR-tree $\mathcal T$ bottom-up, starting with the
  Q-nodes.  If a Q-node $\mu$ corresponds to an edge belonging to a
  cycle $C$, then $\cyc(\mu)$ contains only $C$ and $\bel(\eps) = C$
  for the virtual edge in $\mu$.  If the edge corresponding to $\mu$
  is not contained in a cycle, then $\cyc(\mu)$ is empty and
  $\bel(\eps) = \bot$.  Furthermore, a Q-node cannot be the root of
  any induced subtree $\proj{\mathcal T}{C}$ as we chose as the root
  of $\mathcal T$ a Q-node corresponding to an edge not contained in
  any of the cycles.  Now consider an inner node $\mu$.  We first
  process the virtual edges in $\skel(\mu)$ not belonging to the
  parent of $\mu$.  Let $\eps$ be such a virtual edge corresponding to
  the child $\mu'$ of $\mu$ and let $\eps'$ be the virtual edge in
  $\skel(\mu')$ corresponding to its parent $\mu$.  Then $\eps$
  belongs to a cycle induced by $C$ if and only if~$\eps'$ does, thus
  we set $\bel(\eps) = \bel(\eps')$.  Moreover, if $\bel(\eps) \not=
  \bot$ we need to add the cycle $\bel(\eps)$ to $\cyc(\mu)$ if it was
  not already added.  Whether $\bel(\eps)$ is already contained in
  $\cyc(\mu)$ can be tested in constant time as follows.  We define a
  timestamp $t$, increase $t$ every time we go to the next node in
  $\mathcal T$ and we store the current value of $t$ for a cycle added
  to $\cyc(\mu)$.  Then a cycle $C$ was already added to $\cyc(\mu)$
  if and only if the timestamp stored for $C$ is equal to the current
  timestamp $t$.  Thus, processing all virtual edges in $\skel(\mu)$
  not corresponding to the parent of $\mu$ takes time linear in the
  size of $\skel(\mu)$.  Let now $\eps = \{s, t\}$ be the virtual edge
  corresponding to the parent of $\mu$.  If $\bel(\eps') = \bot$ for
  all virtual edges $\eps'$ incident to $s$, then $\eps$ cannot be
  contained in a cycle induced by any of the cycles in $\mathcal C$.
  Otherwise, there are two possibilities.  There is a cycle $C \in
  \mathcal C$ such that $\bel(\eps_1) = C$ for exactly one virtual
  edge $\eps_1$ incident to $s$ or there are two such edges $\eps_1$
  and $\eps_2$ with $\bel(\eps_1) = \bel(\eps_2) = C$.  In the former
  case the edges belonging to $C$ in $\skel(\mu)$ form a path from $s$
  to $t$, thus the edge $\eps$ corresponding to the parent also
  belongs to $C$ and we set $\bel(\eps) = C$.  In the latter case $s$
  is contained in the cycle $C$ but the virtual edge does not belong
  to $C$ and we set $\bel(\eps) = \bot$.  This takes time linear in
  the degree of $s$ in $\skel(\mu)$ and hence lies in $\mathcal
  O(|\skel(\mu)|)$.  It remains to set $\roo(\proj{\mathcal T}{C}) =
  \mu$ for every cycle $C$ inducing the subtree $\proj{\mathcal T}{C}$
  having $\mu$ as root.  The tree $\proj{\mathcal T}{C}$ has $\mu$ as
  root if and only if $C$ is contained as cycle $\kappa$ in $\mu$ but
  the virtual edge $\eps$ in $\skel(\mu)$ corresponding to the parent
  of $\mu$ is not contained in $\kappa$.  Thus we have to set
  $\roo(\proj{\mathcal T}{C}) = \mu$ for all cycles $C$ in $\cyc(\mu)$
  except for $\bel(\eps)$.  Note that this again consumes time linear
  in the size of $\skel(\mu)$ since the number of cycles that are
  cycles in $\mu$ is in $\mathcal O(|\skel(\mu)|)$.  Due to the fact
  that the SPQR-tree $\mathcal T$ has linear size this yields an
  overall linear running time.
\end{proof}

In the second phase we want to compute $\high(\mu)$ for each of the
nodes in $\mathcal T$.  We obtain the following lemma.

\begin{lemma}
  \label{lem:computing-high}
  Let $G$ be a biconnected planar graph with SPQR-tree~$\mathcal T$
  containing the disjoint cycles~$\mathcal C$.  For every node $\mu$
  in $\mathcal T$ the edge $\high(\mu)$ can be computed in linear
  time.
\end{lemma}
\begin{proof}
  We make use of the fact that $\bel(\eps)$ is already computed for
  every virtual edge $\eps$ in each of the skeletons, which can be
  done in linear time due to Lemma~\ref{lem:induced-trees-comp-data}.
  Note that an edge $\{\mu, \mu'\}$ in the SPQR-tree $\mathcal T$
  (where $\mu$ is the parent of $\mu'$) belongs to the induced subtree
  $\proj{\mathcal T}{C}$ with respect to the cycle $C \in \mathcal C$
  if and only if $\bel(\eps) = C$ for the virtual edge $\eps$ in
  $\skel(\mu)$ corresponding to the child $\mu'$.  In this case we
  also have $\bel(\eps') = \bel(\eps) = C$ where $\eps'$ is the
  virtual edge in $\skel(\mu')$ corresponding to the parent.  Hence,
  we can compute $\high(\mu)$ for every node $\mu$ in~$\mathcal T$ by
  processing~$\mathcal T$ top-down remembering the latest processed
  edge not belonging to any of the induced subtrees.  This can easily
  be done in linear time.
\end{proof}

In the third phase we compute $\det(\pos_C(C'))$ for every crucial
relative position in linear time.  Moreover, we compute $\contr(\eps)$
for every virtual edge $\eps$ and $\detcyc(\mu)$ for every R-node
$\mu$.  We show the following lemma.

\begin{lemma}
  \label{lem:compute-det-and-contr}
  Let $G$ be a biconnected planar graph with SPQR-tree $\mathcal T$
  containing the disjoint cycles $\mathcal C$.  The node
  $\det(\pos_C(C'))$ for each crucial relative position $\pos_C(C')$,
  the list $\contr(\eps)$ for each virtual edges $\eps$ and the list
  $\detcyc(\mu)$ for each R-nodes $\mu$ can be computed in overall
  linear time.
\end{lemma}
\begin{proof}
  Let $C$ and $C'$ be two cycles such that $\pos_C(C')$ is a crucial
  relative position.  We show how to compute the node
  $\det(\pos_C(C'))$ determining this relative position in constant
  time.  Moreover, if $C'$ is contracted in a virtual edge $\eps$ in
  $\det(\pos_C(C'))$, we append the relative position $\pos_C(C')$ to
  $\contr(\eps)$.  Otherwise, $\det(\pos_C(C'))$ is an R-node
  containing $C$ and $C'$ as cycles and we add $\pos_C(C')$ to
  $\detcyc(\det(\pos_C(C')))$.  Since there are only linearly many
  crucial relative positions this takes only linear time.  Let $\mu =
  \roo(\proj{\mathcal T}{C})$ and $\mu' = \roo(\proj{\mathcal T}{C'})$
  be the roots of the induced trees with respect to $C$ and~$C'$,
  respectively, which are already computed due to
  Lemma~\ref{lem:induced-trees-comp-data}.  We use that the lowest
  common ancestor of a pair of nodes can be computed in constant time
  after a linear-time preprocessing~\cite{ht-fafnca-84, bf-lcapr-00}.
  In particular, let $\lca(\mu, \mu')$ be the lowest common ancestor
  of the two roots.  There are three possibilities.  First, $\lca(\mu,
  \mu')$ is above $\mu$ (Figure~\ref{fig:compute-det-cases}(a)).
  Second, $\lca(\mu, \mu') = \mu = \mu'$
  (Figure~\ref{fig:compute-det-cases}(b)).  And third, $\lca(\mu,
  \mu') = \mu$ lies above $\mu'$
  (Figure~\ref{fig:compute-det-cases}(c--f)).  Note that the first
  case includes the situation where $\mu' = \lca(\mu, \mu')$ lies
  above $\mu$.

  In the first case the cycle $C'$ is contracted in $\mu$ in the
  virtual edge $\eps$ corresponding to the parent of $\mu$, while
  $\mu$ contains $C$ as cycle $\kappa$ not containing the virtual edge
  $\eps$ corresponding to the parent.  Hence, $\mu$ determines
  $\pos_{C}(C')$.  We set $\det(\pos_C(C')) = \mu$ and insert
  $\pos_C(C')$ into $\contr(\eps)$.  In the second case $C$ and $C'$
  are both cycles in $\mu = \mu'$, hence $\mu$ determines
  $\pos_{C}(C')$.  We set $\det(\pos_C(C')) = \mu$ and insert
  $\pos_C(C')$ into $\detcyc(\mu)$ since $\skel(\mu)$ contains $C$ and
  $C'$ as cycles.

  \begin{figure}
    \centering
    \includegraphics{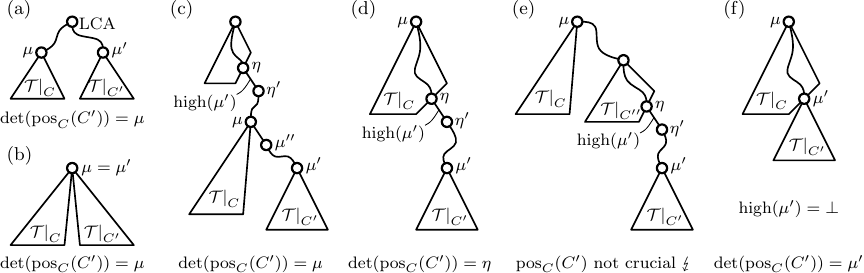}
    \caption{The cases that occur in the proof of
      Lemma~\ref{lem:compute-det-and-contr}.}
    \label{fig:compute-det-cases}
  \end{figure}

  In the third case the node determining $\pos_{C}(C')$ lies somewhere
  on the path from $\mu$ down to~$\mu'$.  In this situation
  $\high(\mu')$ comes into play and we distinguish several cases.  We
  first assume that $\high(\mu') \not= \bot$.  Let $\{\eta, \eta'\} =
  \high(\mu')$ be the highest edge in $\mathcal T$ on the path from
  $\mu'$ to the root that is reachable without using an edge in any of
  the induced trees, as computed by Lemma~\ref{lem:computing-high}.
  Let $\eta$ be the parent of $\eta'$.  We claim that either $\mu$ or
  $\eta$ determines the crucial relative position
  $\pos_C(C')$. 

  More precisely, if $\eta$ lies above or is equal to $\mu$
  (Figure~\ref{fig:compute-det-cases}(c)), then the child $\mu''$ of
  $\mu$ on the path from $\mu'$ to $\mu$ does not contain $C$ as a
  cycle.  Otherwise the edge $\{\mu, \mu''\}$ would have been
  contained in $\proj{\mathcal T}{C}$, which is a contradiction to the
  definition of $\high(\mu')$.  Thus, $C'$ is contracted in the
  virtual edge $\eps$ in $\skel(\mu)$ corresponding to the child
  $\mu''$ that is not contained in the cycle induced by $C$, implying
  that $\mu$ determines $\pos_C(C')$.  In this case we set
  $\det(\pos_C(C')) = \mu$.  Moreover, we want to insert the crucial
  relative position $\pos_C(C')$ into $\contr(\eps)$.  Unfortunately,
  we cannot determine the virtual edge $\eps$ belonging to the child
  $\mu''$ in constant time.  We handle that problem by storing a
  temporary list $\temp(\mu)$ for the node $\mu$ and insert
  $\pos_C(C')$ into this list.  After we have processed all crucial
  relative positions, we process $\mathcal T$ bottom-up, building a
  union-find data structure by taking the union of $\mu$ with all its
  children after processing $\mu$.  Thus, when processing~$\mu$, we
  can simply traverse the list $\temp(\mu)$ once, find for every
  crucial relative position $\pos_C(C')$ the virtual edge $\eps$
  containing $C'$ by finding $\roo(\proj{\mathcal T}{C'})$ in the
  union-find data structure and then add $\pos_C(C')$ to
  $\contr(\eps)$.  Note that this takes overall linear time, because
  the union-find data structure consumes amortized constant time per
  operation since the union-operations we apply are known in
  advance~\cite{gt-ltascdsu-85}.

  In the second case~$\eta$ lies below $\mu$, where $\high(\mu') =
  \{\eta, \eta'\}$.  We claim that $\eta$ contains~$C$ as a cycle and
  $C'$ contracted in the virtual edge $\eps'$ in $\skel(\eta)$
  corresponding to the child~$\eta'$, as depicted in
  Figure~\ref{fig:compute-det-cases}(d).  By definition of
  $\high(\mu')$ there is a cycle $C''$ that is contained as a cycle
  in~$\eta$ and in the parent of~$\eta$ but not in $\eta'$.  We show
  that $C'' = C$ or $\pos_C(C')$ is not a crucial relative position.
  Assume $C'' \not= C$; see Figure~\ref{fig:compute-det-cases}(e).  In
  $\skel(\eta)$ the cycle $C'$ is contracted in the virtual
  edge~$\eps'$ corresponding to the child~$\eta'$, whereas $C$ is
  contracted in the virtual edge $\eps$ corresponding to the parent of
  $\eta$.  Since~$C''$ is a cycle in $\eta$ and in its parent, it
  induces a cycle in $\skel(\eta)$ containing~$\eps$.  Consider a
  path~$\pi$ from~$C'$ to $C$ in the graph $G$.  Then $\pi$ contains
  one of the poles of $\skel(\eta)$ and hence contains a vertex
  in~$C''$.  Thus the relative position $\pos_C(C')$ is not crucial,
  which is a contradiction.  Hence we can simply set $\det(\pos_C(C'))
  = \eta$ and append $\pos_C(C')$ to the list $\contr(\eps')$.

  Finally, $\high(\mu')$ may be not defined, that is $\high(\mu') =
  \bot$ since the edge connecting $\mu'$ to its parent is already
  contained in one of the induced cycle trees.  With a similar
  argument as before, this induced tree is $\proj{\mathcal T}{C}$,
  belonging to the cycle $C$, as depicted in
  Figure~\ref{fig:compute-det-cases}(f).  Thus $\mu'$ contains $C$ and
  $C'$ as cycles and we set $\det(\pos_C(C')) = \mu'$ and add
  $\pos_C(C')$ to the list $\detcyc(\mu')$. This concludes the proof.
\end{proof}

In the fourth and last phase we process the SPQR-tree $\mathcal T$
once more to finally compute the PR-node constraints restricted to the
crucial relative positions.  We obtain the following lemma.

\begin{lemma}
  \label{lem:computing-pr-node-constr}
  Let $G$ be a biconnected planar graph.  The PR-node constraints
  restricted to the crucial relative positions can be computed in
  linear time.
\end{lemma}
\begin{proof}
  We process each node in the SPQR-tree $\mathcal T$ of $G$ once,
  consuming time linear in the size of its skeleton plus some
  additional costs that sum up to the number of crucial relative
  positions in total.  Let $\mu$ be a node in $\mathcal T$.  If $\mu$
  is not contained in any induced tree $\proj{\mathcal T}{C}$, it does
  not determine any relative position at all.  Thus assume there is at
  least one cycle that is a cycle in $\mu$.  If $\mu$ is a P-node,
  $\skel(\mu)$ consists of $\ell$ parallel virtual edges $\eps_1,
  \dots, \eps_\ell$ and we can assume without loss of generality that
  the cycle $C$ induces in $\skel(\mu)$ the cycle $\kappa$ consisting
  of the two virtual edges $\eps_1$ and~$\eps_2$.  For every crucial
  relative position $\pos_{C}(C')$ that is determined by $\mu$ there
  is a virtual edge $\eps \in \{\eps_3, \dots, \eps_\ell\}$ containing
  $C'$ in the list $\contr(\eps)$, which is already computed due to
  Lemma~\ref{lem:compute-det-and-contr}.  Hence, the PR-node
  constraints stemming from $\mu$ can be computed by processing each
  of these lists $\contr(\eps)$, setting $\pos_{C}(C') =
  \pos_{C}(C'')$ for any two cycles $C'$ and $C''$ appearing
  consecutively in $\contr(\eps)$.  The time-consumption is linear in
  the size of $\skel(\mu)$ plus the number of crucial relative
  positions determined by $\mu$.

  If $\mu$ is an R-node, it may contain several cycles as a cycle, all
  of them stored in the list $\cyc(\mu)$ due to
  Lemma~\ref{lem:induced-trees-comp-data}.  Every crucial relative
  position $\pos_C(C')$ determined by $\mu$ is either contained in the
  lists $\contr(\eps)$ for a virtual edge $\eps$ in $\skel(\mu)$ or in
  $\detcyc(\mu)$ if $C$ and $C'$ are both cycles in $\mu$
  (Lemma~\ref{lem:compute-det-and-contr}).  We first carry the
  relative positions in $\detcyc(\mu)$ over to the corresponding
  cycles.  More precisely, we define a list $\detcyc(C')$ for every
  cycle~$C'$ in $\cyc(\mu)$ and insert a crucial relative position
  $\pos_C(C')$ into it, if it is contained in $\detcyc(\mu)$.  This
  can obviously be done consuming time linear in the size of
  $\detcyc(\mu)$.  Afterwards, we start by fixing the embedding of
  $\skel(\mu)$ and pick an arbitrary vertex $v_0$ in $\skel(\mu)$.
  For each cycle $C$ contained as cycle $\kappa$ in $\skel(\mu)$ we
  define a variable $\side(C)$ and initialize it with the value
  ``left'' or ``right'', depending on which side $v_0$ lies \todo{Pag
    20, line 18:}with respect to $\kappa$ in the chosen embedding of
  $\skel(\mu)$, or with the value ``on'' if $v_0$ is contained in~$C$.
  Due to Lemma~\ref{lem:induced-trees-comp-data} we know for every
  edge $\eps$ to which cycle it belongs (or that it does not belong to
  a cycle at all).  Thus $\side(C)$ can be easily computed for every
  cycle $C$ that is a cycle in $\skel(\mu)$ consuming time linear in
  $|\skel(\mu)|$ by traversing $\skel(\mu)$ once, starting at $v_0$.

  To sum up, each crucial relative position $\pos_C(C')$ determined by
  $\mu$ is either contained in $\contr(\eps)$ if $C'$ is contracted in
  $\eps$ or in $\detcyc(C')$ if $C'$ is a cycle in $\skel(\mu)$.
  Moreover, for each cycle~$C$ the value of $\side(C)$ describes on
  which side of $C$ the chosen start-vertex $v_0$ lies with respect to
  a chosen orientation of $\skel(\mu)$.  We now want to divide the
  crucial relative positions determined by~$\mu$ into two lists {\sc
    Left} and {\sc Right} depending on which value they have with
  respect to the chosen embedding.  If this is done, the PR-node
  constraints stemming from $\mu$ restricted to the crucial relative
  positions can be computed by simply processing these two lists once.
  To construct the lists {\sc Left} and {\sc Right}, we make a
  DFS-traversal in $\skel(\mu)$ such that each virtual edge is
  processed once.  More precisely, when we visit an edge $\{u, v\}$
  (starting at $u$), then $v$ is either an unvisited vertex and we
  continue the traversal from $v$ or $v$ was already visited, then we
  go back to $u$.  If all virtual edges incident to the current vertex
  $u$ were already visited, we do a back-tracking step, i.e., we go
  back to the vertex from which we moved to $u$.  Essentially, a
  normal step consists of three phases, leaving the current vertex
  $u$, traveling along the virtual edge $\{u, v\}$, and finally
  arriving at $v$ or back at $u$.  In a back-tracking step we have
  only two phases, namely leaving the current vertex~$u$ and arriving
  at its predecessor.  During the whole traversal we keep track of the
  sides $\side(\cdot)$.  More precisely, when leaving a vertex $u$
  that was contained in a cycle $C$ we may have to update $\side(C)$
  if the target-vertex $v$ is not also contained in $C$.  On the other
  hand, when arriving at a vertex $v$ contained in a cycle $C$ we have
  to set $\side(C) = \text{``on''}$.  Since such an update has to be
  done for at most one cycle we can keep track of the sides in
  constant time per operation and thus in overall linear time.  Now it
  is easy to compute the values of the crucial relative positions
  determined by $\mu$ with respect to the currently chosen embedding.
  While traveling along a virtual edge $\eps = \{u, v\}$ we process
  $\contr(\eps)$.  For a crucial relative position $\pos_C(C')$
  contained in $\contr(\eps)$ we know that $C'$ is contracted in
  $\eps$.  Thus, in the chosen embedding the value of $\pos_C(C')$ is
  the current value of $\side(C)$ and we can insert $\pos_C(C')$ into
  the list {\sc Left} or {\sc Right} depending on the value of
  $\side(C)$.  This takes linear time in the number of crucial
  relative positions contained in $\contr(\eps)$.  We deal with the
  crucial relative positions contained in one of the lists
  $\detcyc(C')$ in a similar way.  Every time we reach a vertex $v$
  contained in a cycle $C'$ we check whether this is the first time we
  visit the cycle $C'$.  If it is the first time, we insert every
  crucial relative positions $\pos_C(C')$ contained in $\detcyc(C')$
  into one of the lists {\sc Left} or {\sc Right}, depending on the
  current value of $\side(C)$.  Clearly the whole traversal takes
  linear time in the size of $\skel(\mu)$ plus linear time in the
  number of crucial relative positions determined by $\mu$.  Moreover,
  we obviously obtain the PR-node constraints restricted to the
  crucial relative positions stemming form $\mu$ by processing each of
  the lists {\sc Left} and {\sc Right} once, obtaining an equality
  constraint for positions that are adjacent in the lists and
  additionally a single inequality for a pair of positions, one
  contained in {\sc Left} and the other in {\sc Right}, unless one of
  them is empty.
\end{proof}

\begin{corollary}
  The CC-tree $\mathcal T_{\mathcal C}$ of a biconnected planar graph $G$ can be
  computed in linear time.
\end{corollary}

It remains to extend the described algorithm to the case where $G$ is
not necessarily biconnected.  More precisely, we need to show how to
compute the extended PR-node constraints and the cutvertex constraints
in linear time.  This is done in the proof of the following theorem.

\begin{theorem}
  \label{thm:compute-cc-trees-lin-time}
  The CC-tree $\mathcal T_{\mathcal C}$ of a connected planar graph $G$ can be
  computed in linear time.
\end{theorem}
\begin{proof}
  As before, the underlying C-tree can be easily computed in linear
  time.  For a fixed block $B$ we have the SPQR-tree $\mathcal T$ and
  for a cycle $C$ in $B$ the induced tree $\proj{\mathcal T}{C}$ can
  be defined as before.  Obviously,
  Lemma~\ref{lem:induced-trees-comp-data} can be used as before to
  compute $\cyc(\mu)$ for every node $\mu$, $\bel(\eps)$ for every
  virtual edge and $\roo(\proj{\mathcal T}{C})$ for every induced
  subtree in linear time.  Moreover, the edge $\high(\mu)$ in
  $\mathcal T$ can be computed for every node $\mu$ as in
  Lemma~\ref{lem:computing-high}.  For the computation of
  $\det(\pos_C(C'))$ for every crucial relative position $\pos_C(C')$
  and $\contr(\eps)$ for every virtual edge $\eps$, we cannot directly
  apply Lemma~\ref{lem:compute-det-and-contr} since the cycles $C$ and
  $C'$ may be contained in different blocks.  Thus, before we can
  compute $\det(\pos_C(C'))$, we need to find out whether $C$ and $C'$
  are in the same block, which can be done by simply storing for every
  cycle a pointer to the block containing it.  For the case that~$C$
  and $C'$ are contained in the same block $\det(\pos_C(C'))$ can be
  computed as before and $\pos_C(C')$ can be inserted into the list
  $\contr(\eps)$ for some $\eps$ if necessary.  For the case that $C$
  and~$C'$ are contained in different blocks $B$ and $B'$, we need to
  find the unique cutvertex $v$ in $B$ that separates $B$ and $B'$.
  This can be done in overall linear time by computing the BC-tree and
  using an approach combining the lowest common ancestor and
  union-find data structure similar as in the proof of
  Lemma~\ref{lem:compute-det-and-contr}.  

  If the resulting cutvertex $v$ is not contained in $C$, we can treat
  the cutvertex~$v$ as if it was the cycle $C'$ and use the same
  algorithm as in Lemma~\ref{lem:compute-det-and-contr} to compute
  $\det(\pos_C(C'))$ and append $\pos_C(C')$ to $\contr(\eps)$ for
  some $\eps$ if necessary.  If $v$ is contained in $C$, then the
  crucial relative position $\pos_C(C')$ is not determined by any node
  in any SPQR-tree at all, but by the embedding of the blocks around
  the cutvertex $v$.  Thus there are no extended PR-node constraints
  restricting $\pos_C(C')$.  Finally, $\det(\pos_C(C'))$ can be
  computed in overall linear time for every crucial relative position
  $\pos_C(C')$ that is determined by a node in the SPQR-tree of the
  block containing~$C$.  Moreover, for a node $\mu$ in the SPQR-tree
  of the block $B$ containing $C$ every virtual edge $\eps$ has a list
  $\contr(\eps)$ containing all crucial relative positions
  $\pos_C(C')$ that are determined by~$\mu$ and for which either $C'$
  is contracted in $\eps$ or belongs to a different block $B'$ and is
  connected to~$B$ via a cutvertex contained in the expansion graph
  $\expan(\eps)$.  With these information the extended PR-node
  constraints can be computed exactly the same as the PR-node
  constraints are computed in
  Lemma~\ref{lem:computing-pr-node-constr}.

  It remains to compute the cutvertex constraints restricted to the
  crucial relative positions.  As mentioned above, we can compute a
  list of crucial relative positions $\pos_C(C_1), \dots,
  \pos_C(C_\ell)$ determined by the embedding of the blocks around a
  cutvertex $v$ contained in $C$ in linear time.  We then process this
  list once, starting with $\pos_C(C_1)$.  We store $\pos_C(C_1)$ as
  reference position for the block~$B_1$ containing $C_1$.  Now, when
  processing $\pos_C(C_i)$, we check whether the block~$B_i$
  containing~$C_i$ already has a reference position $\pos_C(C_j)$
  assigned to it.  In this case we set $\pos_C(C_i) = \pos_C(C_j)$,
  otherwise we set $\pos_C(C_i)$ to be the reference position.  This
  obviously consumes overall linear time and computes the cutvertex
  constraints restricted to the crucial relative positions.
\end{proof}

\subsubsection*{Intersecting CC-Trees in Linear Time}
\label{sec:inters-cc-trees-lin-time}

Due to Theorem~\ref{thm:intersection-is-intersection} we can test
whether two graphs $\1G$ and $\2G$ with common graph $C$ consisting of
a set of disjoint cycles have a {\sc SEFE} by computing the
CC-trees $\1{\mathcal T_{\mathcal C}}$ and $\2{\mathcal T_{\mathcal C}}$ of $\1G$ and $\2G$,
respectively, which can be done in linear time due to
Theorem~\ref{thm:compute-cc-trees-lin-time}.  Then the intersection
$\mathcal T_{\mathcal C}$ of $\1{\mathcal T_{\mathcal C}}$ and $\2{\mathcal T_{\mathcal C}}$ represents
exactly the possible embeddings of the common graph $G$ in a {\sc
  SEFE}.  It remains to show that the intersection can be computed in
linear time.

\begin{theorem}
  \label{thm:intersection-in-lin-time}
  The intersection of two CC-trees can be computed in linear time.
\end{theorem}
\begin{proof}
  Let $\1{\mathcal T_{\mathcal C}}$ and $\2{\mathcal T_{\mathcal C}}$
  be two CC-trees on a set~$\mathcal C$ of cycles.  We start with
  $\mathcal T_{\mathcal C} = \1{\mathcal T_{\mathcal C}}$ and show how
  to compute the common-face and crucial-position constraints in
  overall linear time.  For the crucial-position constraints we
  essentially only show how to find for each crucial relative position
  in $\2{\mathcal T_{\mathcal C}}$ a crucial relative position in
  $\mathcal T_{\mathcal C}$ corresponding to it.  Computing the
  crucial-position constraints is then easy.  We root $\mathcal
  T_{\mathcal C}$ at an arbitrary vertex and again use that the lowest
  common ancestor of two vertices in~$\mathcal T_{\mathcal C}$ can be
  computed in constant time~\cite{ht-fafnca-84, bf-lcapr-00}.  For
  every edge $\2e = \{C_1, C_2\}$ in $\2{\mathcal T_{\mathcal C}}$ we
  obtain a path in $\mathcal T_{\mathcal C}$ from $C_1$ to the lowest
  common ancestor of $C_1$ and $C_2$ and further to $C_2$.  We
  essentially process these two parts of the path separately with some
  additional computation for the lowest common ancestor.  We say that
  the parts of the paths \emph{belong} to the \emph{half-edges}
  $\2{e_1}$ and $\2{e_2}$, respectively.  We use the following
  data structure.  For every cycle $C$ there is a list $\en(C)$
  containing all edges in~$\2{\mathcal T_{\mathcal C}}$ whose
  endpoints have $C$ in $\mathcal T_{\mathcal C}$ as lowest common
  ancestor.  This list can be computed for every cycle in overall
  linear time.  We then process $\mathcal T_{\mathcal C}$ bottom up,
  saving for the cycle $C$ we currently process a second list
  $\curr(C)$ containing all the half-edges in $\2{\mathcal T_{\mathcal
      C}}$ whose paths contain $C$.  This can be done in overall
  linear time by ensuring that every half-edge $\2{e_i}$ ($i = 1, 2$)
  is contained in at most one list $\curr(C)$ at the same time.  Then
  $\2{e_i}$ can be removed from this list in constant time by storing
  for~$\2{e_i}$ pointers to the previous and to the next element in
  that list, denoted by $\prev(\2{e_i})$ and $\nex(\2{e_i})$.
  Additionally, we build up the following union-find data structure.
  Every time we have processed a cycle $C$, we union $C$ with all its
  children in $\mathcal T_{\mathcal C}$.  Thus, when processing~$C$,
  this data structure can be used to find for every cycle in the
  subtree below $C$ the child of $C$ it belongs to.  Note that again
  this version of the union-find data structure consumes amortized
  constant time per operation since the sequence of union operations
  is known in advance~\cite{gt-ltascdsu-85}.  Before starting to
  process $\mathcal T_{\mathcal C}$, we process~$\2{\mathcal
    T_{\mathcal C}}$ once and for every edge~$\2e = \{C_1, C_2\}$ we
  insert the half-edges~$\2{e_1}$ and~$\2{e_2}$ to the lists
  $\curr(C_1)$ and $\curr(C_2)$, respectively.  While
  processing~$\mathcal T_{\mathcal C}$ bottom up the following
  invariants hold at the moment we start to process $C$.
  \begin{compactenum}
  \item The list $\curr(C)$ contains all half-edges starting at $C$.
  \item For every child $C'$ of $C$ the list $\curr(C')$ contains the
    half-edge $\2{e_i}$ if and only if the path belonging to it
    contains $C$ and $C'$.
  \item Every half-edge $\2{e_i}$ is contained in at most one list
    $\curr(C)$, and $\prev(\2{e_i})$ and $\nex(\2{e_i})$ contain the
    previous and next element in that list, respectively.
  \end{compactenum}
  When we start processing a leaf $C$ the invariants are obviously
  true.  To satisfy invariant~2. for the parent of $C$ we have to
  ensure that all half-edges ending at $C$ are removed form the list
  $\curr(C)$.  Since there are no half-edges ending in a leaf, we
  simply do nothing.  Invariants~1. and~3. obviously also hold for the
  parent of $C$.

  Let $C$ be an arbitrary cycle and assume that the invariants are
  satisfied.  To ensure that invariant~2. holds for the parent of
  $C$, we need to build a list of all half-edges whose paths
  contain~$C$ and do not end at $C$.  Since invariants~1. and~2. hold
  for $C$ this are exactly the half-edges contained in $\curr(C)$ plus
  the half-edges contained in $\curr(C')$ for each of the children
  $C'$ of $C$ that are not ending at $C$.  Note that a half-edge may
  also start at $C$ and end at $C$.  This is the case if the
  corresponding edge connects $C$ with another cycle $C''$ such that
  the lowest common ancestor of~$C$ and $C''$ is $C$.  We first
  process the list $\en(C)$ containing the edges ending at $C$; let
  $\2e$ be an edge in $\en(C)$.  The two half-edges $\2{e_1}$ and
  $\2{e_2}$ belonging to $\2e$ are contained in the lists $\curr(C_1)$
  and $\curr(C_2)$, where $C_1$ and $C_2$ are different cycles and
  each of them is either $C$ or a child of $C$.  We remove $\2{e_i}$
  form $\curr(C_i)$ for~$i=1,2$.  This can be done by setting the
  pointers $\nex(\prev(\2{e_i})) = \nex(\2{e_i})$ and
  $\prev(\nex(\2{e_i})) = \prev(\2{e_i})$, taking constant time per
  edge since each half-edge is contained in at most one list due to
  invariant~3.  Afterwards, for every child $C'$ of $C$, we append
  $\curr(C')$ to $\curr(C)$ and empty the list $\curr(C')$ afterwards,
  to ensure that invariant~3 remains satisfied.  This takes constant
  time per child and thus overall time linear in the degree of $C$.
  Obviously this satisfies all invariants for the parent of $C$.
  Furthermore, we consume time linear in the degree of~$C$ plus time
  linear in the number of half-edges ending at $C$.  However, every
  half-edge ends exactly once yielding overall linear time.

  Now it is easy to compute the common-face and crucial-position
  constraints while processing~$\mathcal T_{\mathcal C}$ as described
  above.  Essentially, when processing $C$, we compute all the
  constraints concerning the relative position of other cycles with
  respect to $C$.  In particular, we need to add common-face
  constraints if two half-edges end at $C$ and if the path belonging
  to a half-edge contains $C$ in its interior.  Furthermore, we find a
  corresponding crucial relative position for every half-edge starting
  at $C$.  \todo{Pag 23, lines 21-24:}Let $\2e = \{C_1, C_2\}$ be an
  edge whose half-edges end at $C$.  There are two different cases.
  First, one of the cycles $C_i$ (for $i = 1, 2$) is $C$ (its half
  edge starts and ends at $C$).  Then the other cycle (whose half-edge
  only ends at $C$) is contained in a subtree with root $C'$, where
  $C'$ is a child of $C$.  Second, $C_1$ and $C_2$ are contained in
  the subtrees with roots $C_1'$ and $C_2'$, respectively, where
  $C_1'$ and $C_2'$ are different children of $C$.  In this case, both
  half-edges end at $C$.  We consider the second case first.  Then
  $C_1'$ and $C_2'$ can be found in amortized constant time by finding
  $C_1$ and $C_2$ in the union-find data structure.  The equation
  $\pos_C(C_1') = \pos_C(C_2')$ is exactly the common-face constraint
  at the cycle $C$ stemming from the edge $\2e$.  In the second case
  we can again find the child $C'$ in constant time.  Assume without
  loss of generality that $C_1 = C$ and $C_2$ is contained in the
  subtree having $C'$ as root.  Then $\pos_C(C')$ is the crucial
  relative position in $\mathcal T_{\mathcal C}$ corresponding to the
  crucial relative position $\pos_{C}(C_2)$ in $\2{\mathcal
    T_{\mathcal C}}$.  The half-edges containing $C$ in its interior
  are exactly the half-edges contained in one of the lists $\curr(C')$
  for a child $C'$ of $C$ whose path does not end at $C$.  Thus, for
  the parent $C''$ of $C$ we have to add the common-face constraint
  $\pos_C(C') = \pos_C(C'')$ if and only if the list $\curr(C')$ is
  not empty after deleting all half-edges in $\en(C)$.  These
  additional computations obviously do not increase the running time
  and hence the common-face and crucial-position constraints can be
  computed in overall linear running time.
\end{proof}

Theorems~\ref{thm:cycle-tree-is-rep},~\ref{thm:intersection-is-intersection},~\ref{thm:compute-cc-trees-lin-time}
and~\ref{thm:intersection-in-lin-time} directly yield the following results.

\begin{theorem}
  {\sc Simultaneous Embedding with Fixed Edges} can be solved in
  linear time if the common graph consists of disjoint cycles.
\end{theorem}

\begin{theorem}
  {\sc SEFE} can be solved in linear time for the case of $k$ graphs
  $\1G, \dots, \k G$ all intersecting in the same common graph $G$
  consisting of disjoint cycles.
\end{theorem}

\section{Connected Components with Fixed Embedding}
\label{sec:extension}

In this section we show how the previous results can be extended to
the case that the common graph has several connected components, each
of them with a fixed planar embedding.  Again, we first consider the
case of a single graph $G$ containing $\mathcal C$ as a subgraph,
where in this case $\mathcal C$ is a set of connected components
instead of a set of disjoint cycles.  First note that the relative
position $\pos_C(C')$ of a component $C'$ with respect to another
component $C$ can be an arbitrary face of $C$.  Thus, the choice of
the relative positions is no longer binary and a set of inequalities
on the relative positions would lead to a coloring problem in the
conflict graph, which is $\mathcal{NP}$-hard in general.  However,
most of the constraints between relative positions are equations, in
fact, all inequalities stem from R-nodes in the SPQR-tree of $G$ (or
of the SPQR-tree of one of the blocks in~$G$).  Fortunately, if a
relative position is determined by an R-node, there are only two
possibilities to embed this R-node.  Thus, the possible values for
the relative position is restricted to two faces, yielding a binary
decision.  Note that in general the possible values for $\pos_C(C')$
are not all faces of~$C$, even if $\pos_C(C')$ is not determined by an
R-node but by a P-node or by the embedding around a cutvertex.

Thus, we obtain for each relative position a set of possible faces as
values and additionally several equations and inequalities, where
inequalities only occur between relative positions with a binary
choice.  These conditions can be modeled as a conflict graph where
each node represents a relative position with some allowed colors
(faces) and edges in this conflict graph enforce both endvertices to
be either colored the same or differently.  In the case of the problem
{\sc SEFE} each of the graphs yields such a conflict graph.  These
conflict graphs can be easily merged by intersecting for each relative
position the sets of allowed colors (faces).  Then a simultaneous
embedding can be constructed by first iteratively contracting edges
requiring equality, intersecting the possible colors of the involved
nodes.  If the resulting graph contains a node with the empty set as
choice for the color, then no simultaneous embedding exists.
Otherwise, we have to test whether each connected component in the
remaining graph can be colored such that adjacent nodes have different
colors, which can be done efficiently since such a component either
consists of a single node or there are only up to two possible colors
for each connected component left due to the considerations above.

Moreover, the CC-tree can be adapted to work for the case of connected
components with fixed embeddings instead of disjoint cycles, as the
extended PR-node and cutvertex constraints on the crucial relative
positions are still sufficient to imply them on all relative
positions.  We call this tree on connected components the
CC$^\oplus$-tree, standing for \emph{constrained component-tree}.  In
the following we quickly go through the steps we did before in the
case of disjoint cycles and describe the changes when considering
connected components instead.

\paragraph{PR-Node Constraints.}
\label{sec:pr-node-constraints}

Let $G$ be a biconnected planar graph and let $\mathcal C$ be a
subgraph of $G$ consisting of several connected components, each with
a fixed planar embedding.  Let further $C \in \mathcal C$ be one of
the connected components and let $\mu$ be a node in the SPQR-tree
$\mathcal T$ of $G$.  The virtual edges in $\skel(\mu)$ whose
expansion graphs contain parts of the component $C$ induce a connected
subgraph in $\skel(\mu)$.  In the previous case, where the subgraph
consisted of disjoint cycles, this induced subgraph was either a
single edge or a cycle.  In the case that $C$ is an arbitrary
component the induced graph can be an arbitrary connected subgraph of
$\skel(\mu)$.  If it is a single edge, we say that $C$ is
\emph{contracted} in $\mu$, otherwise $C$ is a \emph{component} in
$\mu$.

We obviously obtain that the relative position $\pos_C(C')$ of another
component $C'$ with respect to $C$ is determined by the embedding of
$\skel(\mu)$ if and only if $C$ is a component in $\mu$ and $C'$ is
not contracted in one of the virtual edges belonging to the subgraph
induced by $C$.  Moreover, the embedding of $\skel(\mu)$ is partially
(or completely) fixed by the embedding of $C$ if the induced graph in
$\skel(\mu)$ contains a vertex with degree greater than~2.  More
precisely, consider $\mu$ to be a P-node containing $C$ as a
component.  Then the virtual edges belonging to $C$ have a fixed
planar embedding and each face in this induced graph represents a face
in $C$.  These faces are the possible values for the relative
positions with respect to $C$ that are determined by $\mu$.  The
remaining virtual edges not belonging to $C$ can be added arbitrarily
and thus components contracted in these edges can be put into one of
the possible faces with the restriction that two components contracted
in the same virtual edge have to lie in the same face, that is they
have the same relative position with respect to $C$.  To sum up, we
obtain a set of possible faces of $C$ with respect to $\mu$ and a set
of equations between relative positions of components with respect to
$C$.

For the case that $\mu$ is an R-node, either the embedding of
$\skel(\mu)$ is fixed due to the fact that there exists a component
whose induced graph in $\skel(\mu)$ contains a vertex with degree
greater than~2.  Otherwise, each component is either contracted in
$\mu$ or the induced subgraph is a cycle or a path.  No matter which
case arises, the relative positions determined by $\mu$ are either
completely fixed or there are only two possibilities.  If the
embedding is fixed, the relative positions determined by $\mu$ are
fixed and thus there is no need for additional constraints.
Otherwise, a crucial relative position with respect to $C$ is fixed if
$C$ induces a path in $\skel(\mu)$ and it changes by flipping
$\skel(\mu)$ if $C$ induces a cycle.  For two components $C$ and $C'$
both inducing a cycle in $\skel(\mu)$ this yields a bijection between
the two possible values for relative positions with respect to $C$
determined by $\mu$ and the two possible values for positions with
respect to $C'$.  Thus we can add the equations and inequalities as in
the case of disjoint cycles.

The resulting constraints are again called PR-node constraints.  As
for disjoint cycles we obtain that an embedding of the components
$\mathcal C$ respecting the fixed embeddings for each component can be
induced by an embedding chosen for $G$ if and only if the PR-node
constraints are satisfied.  This directly yields a polynomial-time
algorithm to solve {\sc SEFE} for the case that both graphs are
biconnected and the common graph consists of several connected
components, each having a fixed planar embedding.

\paragraph{Extended PR-Node and Cutvertex Constraints.}
\label{sec:extended-pr-node}

As for cycles the considerations above can be easily extended to the
case that the graph $G$ containing the components $\mathcal C$ is
allowed to contain cutvertices.  For a cutvertex $v$ not contained in
a component $C$, the relative position of $v$ with respect to $C$
determines the relative positions of components attached via $v$,
which again yields the extended PR-node constraints.  If $v$ is
contained in $C$, then the relative position of another component $C'$
with respect to $C$ is determined by the embedding around $v$ if and
only if $v$ splits $C$ from $C'$.  In this case $C'$ can obviously lie
in one of the faces of $C$ incident to $v$.  Fortunately, the
cutvertex constraints do not contain inequalities as they only ensure
that components attached to $v$ via the same block lie in the same
face of $C$.  With these considerations all results form
Section~\ref{sec:disj-cycl-poly-time} can be extended to the case of
components with fixed embedding instead of cycles.  In particular,
{\sc SEFE} can be solved in polynomial time if the common graph
consists of connected components, each with a fixed planar embedding.

\paragraph{CC$^{\boldsymbol\oplus}$-Trees.}
\label{sec:cc-trees-ext}

As mentioned before, the CC-tree can be adapted to represent all
embeddings that can be induced on the set of components $\mathcal C$
by an embedding of the graph $G$, yielding the CC$^\oplus$-tree.  To
this end, each node in the tree represents a component $C \in \mathcal
C$ and the incidence to $C$ of an edge $\{C, C'\}$ in the
CC$^\oplus$-tree represents the choice for the crucial relative
position $\pos_C(C')$.  The possible values are restricted to a subset
of faces of $C$ as described before and there may be some equations
between crucial relative positions with respect to $C$.  Moreover,
there may be inequalities between crucial relative positions even with
respect to different components.  However, if this is the case, then
there are at most two possible choices and we have a bijection between
the possible faces of different components.  As in the proof of
Theorem~\ref{thm:cycle-tree-is-rep}, it follows from the structure of
the underlying C-tree, that relative positions that are not crucial
are determined by a crucial relative position that is determined by
the same P- or R-node or by the same embedding choice around a
cutvertex.  The proof can be easily adapted to the case of components
instead of cycles yielding that satisfying the constraints and
restrictions to a subset of faces for the crucial relative positions
automatically satisfies these conditions for all relative positions.

To be able to solve {\sc SEFE} with the help of CC$^\oplus$-trees, we need to
intersect two CC$^\oplus$-trees such that the result is again a CC$^\oplus$-tree.
Assume as in the case of cycles that we have the two
CC$^\oplus$-trees~$\1{\mathcal T_{\mathcal C}}$ and~$\2{\mathcal T_{\mathcal C}}$.  As before we
start with $\1{\mathcal T_{\mathcal C}}$ and add the restrictions given by
$\2{\mathcal T_{\mathcal C}}$.  More precisely, for every pair $\{C, C'\}$ of
adjacent nodes in $\2{\mathcal T_{\mathcal C}}$ we have to add the common-face
constraints to $\1{\mathcal T_{\mathcal C}}$, that is equations between crucial
relative positions on the path between $C$ and $C'$ in $\1{\mathcal
  T_{\mathcal C}}$ enforcing~$C$ and~$C'$ to share a face.  Moreover, for every
relative position $\pos_C(C')$ that is crucial with respect to
$\2{\mathcal T_{\mathcal C}}$ we have to add the equations and inequalities it is
involved in to the CC$^\oplus$-tree $\1{\mathcal T_{\mathcal C}}$.  As for cycles
$\pos_C(C')$ is in $\1{\mathcal T_{\mathcal C}}$ determined by the crucial
relative position $\pos_C(C'')$, where $C''$ is the first node on the
path from $C$ to $C'$.  We have to do two things.  First, we have to
restrict the possible choices for $\pos_C(C'')$ to those that are
possible for $\pos_C(C')$, which can easily be done by intersecting
the two sets.  Second, the equations and inequalities $\pos_C(C')$ is
involved in have to be carried over to $\pos_C(C'')$.  This can be
done as before by choosing for each crucial relative position in
$\2{\mathcal T_{\mathcal C}}$ the representative in $\1{\mathcal T_{\mathcal C}}$.  For the
resulting intersection $\mathcal T_{\mathcal C}$ it remains to show that every
embedding represented by it is also represented by $\1{\mathcal T_{\mathcal C}}$
and $\2{\mathcal T_{\mathcal C}}$.  The former is clear, the latter can be shown
as in the proof of Theorem~\ref{thm:intersection-is-intersection}.

\paragraph{Efficient Implementation.}

Unfortunately, the constrained component-tree may have quadratic size
in contrast to the constrained cycle-tree, which has linear size.
This comes from the fact that a node $C$ in the CC$^\oplus$-tree may
have linearly many neighbors.  Moreover, each relative position
$\pos_C(C')$ of a neighbor $C'$ of $C$ may have linearly many possible
values, as $C$ may have that many faces.  As these possible values
need to be stored for the edge $\{C, C'\}$ in the CC-tree it has
quadratic size.  On the other hand, it is easy to see that the CC-tree
can be computed in quadratic time.  Moreover, the proof of
Theorem~\ref{thm:intersection-in-lin-time} providing a linear-time
algorithm to intersect CC$^\oplus$-trees can be adapted almost
literally.  The only thing that changes is that additionally the
possible values for $\pos_C(C'')$ and $\pos_C(C')$ need to be
intersected, where $\pos_C(C')$ is a relative position that is crucial
with respect to $\2{\mathcal T_{\mathcal C}}$ and $\pos_C(C'')$ is the
representative in $\1{\mathcal T_{\mathcal C}}$.  Thus, two
CC$^\oplus$-trees can be intersected consuming time linear in the size
of the CC$^\oplus$-trees, that is quadratic time in the size of the
input graphs.  We finally obtain the following theorem.

\begin{theorem}
  {\sc Simultaneous Embedding with Fixed Edges} can be solved in
  quadratic time, if the embedding of each connected component of the
  common graph is fixed.
\end{theorem}

Let $\mathcal T_{\mathcal C}$ be the CC$^\oplus$-tree representing all
embeddings of the components $\mathcal C$ that can be induced by the
graph $G$.  It is worth noting that, although the explicit
representation of $\mathcal T_{\mathcal C}$ may have quadratic size,
it also admits a compact representation of linear size.  The key idea
is the following.  In case there are more than two possible values for
a crucial relative position $\pos_C(C')$, this position is determined
by a P-node or a cutvertex.  Then we can encode the possible values
for $\pos_C(C')$ by pointing to a list that is stored at that P-node
or cutvertex, respectively.  Since this set of values is independent
of $C'$ it is sufficient to store one list for each P-node or
cutvertex.  It is not hard to see that the total size of these lists
is linear.  Moreover, the fast algorithm for computing CC-trees can be
applied with obvious modifications to compute this compact
representation in linear time.  It is, however, unclear whether the
intersection of two or more CC$^\oplus$-trees still admits a compact
representation and whether it can be computed in linear time from the
given compact representations.

\section{Conclusion}
\label{sec:conclusion}

Contrary to the previous results on simultaneous embeddings we focused
on the case where the embedding choice does not consist of ordering
edges around vertices but of placing connected components in relative
positions to one another.  We first showed that generally both input
graphs of an instance of {\sc Simultaneous Embedding with Fixed Edges}
can be always assumed to be connected.  We then showed how to solve
{\sc Simultaneous Embedding with Fixed Edges} in linear time for the
case that the common graph consists of simple disjoint cycles (or more
generally has maximum degree~2).  We further extended the result to a
quadratic-time algorithm solving the more general case where the
embedding of each connected component of the common graph is fixed.
These solutions include a compact and easy to handle data structure,
the CC-tree and CC$^\oplus$-tree, representing all possible
simultaneous embeddings.  Thus, there is hope that the CC-tree and the
CC$^\oplus$-tree are also useful when relaxing the restriction of a
fixed embedding for each component.

\ifdefined\elsevier
\bibliographystyle{elsarticle-num}
\fi
\ifdefined\arXiv
\bibliographystyle{plain}
\fi
\bibliography{DisjCycles}

\end{document}